\documentclass[journal]{IEEEtran}
\usepackage[T1]{fontenc}

\usepackage{multirow}
\usepackage{tablefootnote}

\usepackage{cite}

\usepackage{subfigure}

\usepackage{graphicx}

\usepackage{amsmath}
\usepackage{amsfonts}
\usepackage{amsthm}
\newtheorem{theorem}{Theorem}
\newtheorem{lemma}{Lemma}
\newtheorem{corollary}{Corollary}
\usepackage{bm}

\usepackage{balance}

\usepackage{algorithm}
\usepackage{algorithmicx}
\usepackage{algpseudocode}

\usepackage{array}

\ifCLASSOPTIONcompsoc
  \usepackage[caption=false,font=normalsize,labelfont=sf,textfont=sf]{subfig}
\else
  \usepackage[caption=false,font=footnotesize]{subfig}
\fi

\begin{document}

\title{The Impact of Interference Cognition on the Reliability and Capacity of Industrial Wireless Communications}

\author{Yichen~Guo, Tao~Peng,~\IEEEmembership{Member,~IEEE}, Yujie~Zhao, Yijing~Niu, Wenbo~Wang,~\IEEEmembership{Senior Member,~IEEE}
\thanks{Yichen Guo, Tao Peng, Yujie Zhao, Yijing Niu and Wenbo Wang
are with the Key Laboratory of Universal Wireless Communications,
Ministry of Education, Beijing University of Posts and Telecommunications,
Beijing, 100876, P.R. China (e-mail: \{guoyichen, pengtao, yjzhao, silhouette, wbwang\}@bupt.edu.cn). (Corresponding author: Tao Peng)}
\thanks{This work was supported in part by
the China National Key R\&D Program (No. 2022YFB3303700) and BUPT Excellent Ph.D Students Foundation (No. CX2023240).}}

\markboth{IEEE Transactions on Vehicular Technology,~Vol.~XX, No.~X, XXX~202X}%
{Guo \MakeLowercase{\textit{et al.}}: The Impact of Interference Cognition on the Reliability and Capacity of Industrial Wireless Communications}

\maketitle

\begin{abstract}
      Interference significantly impacts the performance of industrial wireless networks, particularly n severe interference environments with dense networks reusing spectrum resources intensively. Although delicate interference information is often unavailable in conventional networks, emerging interference cognition techniques can compensate this critical problem with possibly different precision. This paper investigates the relationship between precision of interference cognition and system performance. We propose a novel performance analysis framework that quantifies the impact of varying interference information precision on achievable rate.

      Specifically, leveraging the Nakagami-$\bm{m}$ fading channel model, we analytically and asymptotically analyze the average achievable rate in the finite blocklength regime for different  precision levels of signal and interference information. Our findings reveal the critical importance of identifying per-link interference information for achieving optimal performance. Additionally, obtaining instantaneous information is more beneficial for signal links.
\end{abstract}


\IEEEpeerreviewmaketitle

\section{Introduction} \label{intro}

\IEEEPARstart{T}{he} telecommunications industry is thriving, fueled by advancements in mobile technology and a surge in application demands. 5G, with its enhanced mobile broadband (eMBB) capabilities offering significant data rate, latency, and coverage improvements over previous generations \cite{itur.m2083}, has garnered substantial attention.

Moreover, 5G introduces two novel paradigms: ultra-reliable low-latency communication (URLLC) and massive machine-type communication (mMTC). URLLC ensures industry-grade reliability and latency (reliability of $1 - 10^{-5}$ for a 32-byte packet with 1ms latency \cite{3gpp.38.913}), while mMTC enables massive device connectivity \cite{Jiang2021}. These advancements significantly expand the potential of 5G across various industries.

As 5G development progresses, the industry is shifting towards 6G. While still nascent, a consensus is emerging regarding its key features. Building upon 5G, 6G anticipates enhancements to URLLC, which will advance into hyper-reliable and low-latency communication (HRLLC), guaranteeing exceptional reliability and latency performance for mission-critical applications\cite{itur.m2160}, for example, smart manufacturing, industrial Internet of Things (IIoT) and vehicle-to-everything (V2X) communications.

As the telecommunications technology embracing vertical industries with increasing flexibility and adaptability, performance analysis in the finite blocklength (FBL) regime is becoming more and more crucial for understanding the capabilities of industrial wireless communications, as the blocklength is usually constrained by the latency requirement, rendering the assumption of infinite blocklength invalid. In terms of rate analysis, while significant progress has been made in FBL performance analysis, much of this work focuses on ideal scenarios and does not consider interference \cite{Polyanskiy2010, Polyanskiy2011, Yang2014}.

Recent research has attempted to address this limitation by incorporating elements like non-orthogonal multiple access (NOMA), full duplex (FD), relay or multiple-input multiple-output (MIMO) techniques. Zheng \emph{et al.} \cite{Zheng2019} analyze a downlink NOMA system with perfect successive interference cancellation (SIC) decoding, considering only one interference source. Some works employ a similar system model but consider imperfect SIC decoding \cite{Tran2021, Hoang2024}. In contrast, Le \emph{et al.} \cite{Le2022} focus on uplink NOMA systems with perfect SIC decoding and fixed decoding order, which may necessitate further refinements for practical applications. FD and reconfigurable intelligent surfaces (RIS) are often used synergistically to enhance system flexibility in relay systems \cite{Gu2018a}. However, simplified approaches often neglect complex radio conditions and inter-user interference \cite{Gu2018a, Sharma2023}. Whereas in performance analysis for MIMO scenario, additive white Gaussian noise (AWGN) channel is often considered \cite{Ostman2021}, treating interference as Gaussian noise. Besides, these works majorly focus on block error rate (BLER), and achievable rate is calculated by utilizing mean BLER and instantaneous signal-to-interference-plus-noise ratio (SINR). As the instantaneous SINR can barely be obtained in grant-free transmission, these works' practicality in industrial wireless networks is in dispute.

Focusing on rate-related analysis, stochastic network modeling\cite{Zhang2024} and stochastic geometry\cite{Hesham2024, Hesham2024a} have gained prominence in research due to their strong analytical capabilities. However, industrial applications such as autonomous manufacturing present unique challenges: these environments typically feature a massive deployment of heterogeneous devices, including sensors, robotic arms, and programmable logic controllers (PLCs), within densely packed factory spaces. To satisfy stringent quality-of-service (QoS) requirements in high-density scenarios, multi-cell architectures with resource-sharing multiple access points (APs) have emerged as critical infrastructure. Regrettably, conventional stochastic network layouts inherently lack specificity in deployment configurations, a limitation that fundamentally undermines their applicability for precise performance evaluation of networks with defined topologies.

This highlights a critical challenge: existing approaches for performance analysis within FBL regime struggle to accurately capture the complex interference present in real-world multi-cell industrial wireless networks. In this case, severe inter-cell interference (ICI) significantly impacts the system's performance and QoS.

Previous attempts to model ICI have used predictive methods based on aggregated interference \cite{Mahmood2021}, stochastic analysis focusing on generalized network \cite{Zhang2021}, or approximated ICI with its expectation \cite{Zhang2022}. While these approaches offer some insights, they often lack the precision needed to accurately predict performance in complex scenarios.

Our previous work leverages a neural network (NN) \cite{Cao2020} and a non-linear regression algorithm (NLRA) \cite{Peng2022} to model ICI and predict SINR. The NN provides accurate modeling, while the NLRA enhanced it with improved interpretability and computational efficiency, enabling accurate modeling of average power for each interference link.

However, both approaches struggle with the unpredictability of fast fading (FF) and its impact on interference cognition. This paper aims to bridge this gap by investigating the precise role of interference cognition in projecting achievable rate in industrial wireless networks within the FBL regime. The main contributions of this paper can be summarized as follows:

\begin{enumerate}
      \item We propose a theoretical framework for analyzing achievable rate in the FBL regime, considering a wider range of FF distributions and complex ICI in large-scale industrial wireless networks.
      \item We analytically and asymptotically analyze achievable rate with respect to different levels of interference cognition precision. Which is of special significance in grant-free transmission, with instant information unavailable and resources pre-allocated.
      \item Our analysis reveals the crucial role of link-level interference cognition in ensuring reliable service provision. Improved interference cognition precision leads to higher projected achievable rates, particularly when combined with precise signal power knowledge. We also find that the sensitivity of achievable rate to fading severity and QoS requirements is primarily influenced by the accuracy of signal power knowledge.
\end{enumerate}

The rest of the paper is organized as follows: Section~\ref{sysModel} describes the system model; Section~\ref{DCC} elaborates on the structure of the performance analysis framework; Section~\ref{statModel} presents statistical models of SINR for different levels of interference cognition; Section~\ref{perfAnal} conducts rate and asymptotic analyses. Numerical verification and analysis are presented in Section~\ref{Sim}, and conclusions are drawn in Section~\ref{End}.

\section{System Model} \label{sysModel}
\begin{figure}[!t]
      \centering
      \includegraphics[width=3in]{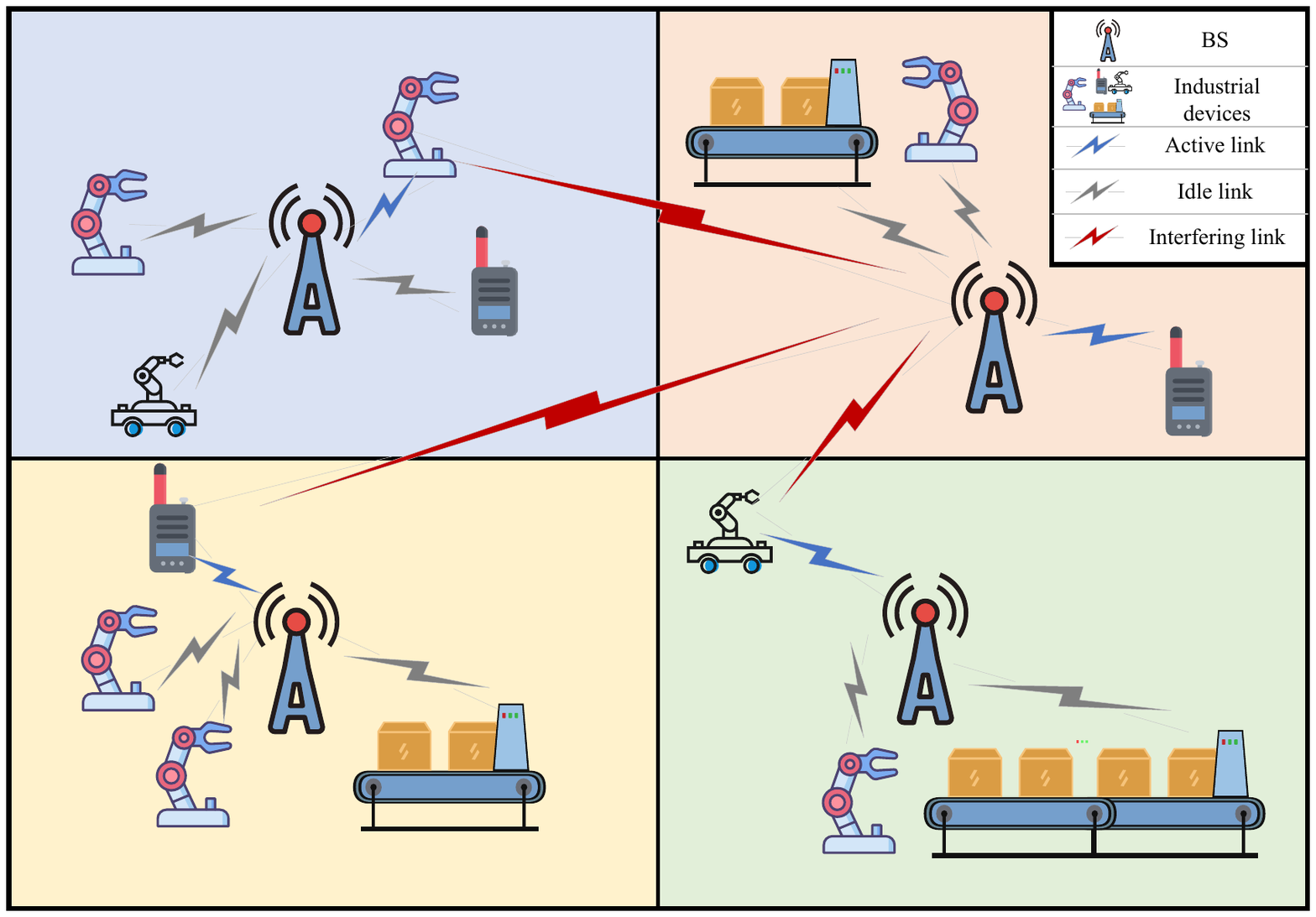}
      \caption{The system model. For brevity, only the interference links of the upper right cell is depicted.}
      \label{fig_sysModel}
\end{figure}
\subsection{System Setup}

This paper analyzes uplink orthogonal multiple access (OMA)-based industrial cellular system for elimination of intra-cell interference, as depicted in Fig. \ref{fig_sysModel}.

For simplicity, we denote the serving base station (BS) as $C_0$\footnote{Since a cell has only one BS, the cell and its BS are represented by the same symbol in this paper.}. The user equipment (UE) of interest at current resource is denoted as $U_0$. The sets $\mathbf{U}=\{U_1, U_2, \dots, U_{|\mathbf{U}|}\}$ and $\mathbf{C}=\{C_{z_1}, C_{z_2}, \dots, C_{z_{|\mathbf{U}|}}\}$ represent the active interfering UEs and their associated BSs at current resource, respectively, with $z_i$ being the index of BS subscribed by $U_i$\footnote{In this context, $z_0 = 0$. Thus, $C_{z_0} = C_0$.}. Furthermore, $\mathbf{A}$ denotes the set of all (active and inactive) interfering UEs, with $\mathbf{U} \subseteq \mathbf{A}$.

\subsection{Channel Model}

The path loss between $U_i$ and $C_{z_j}$ is denoted as $L_{ij}$. The FF for the link between a UE and its BS follows Nakagami-$m$ distribution
\footnote{The power of industrial wireless signal subjects to gamma distribution \cite{Olofsson2016, Kountouris2014}, which corresponds to Nakagami-$m$ distribution in amplitude. Further, Nakagami-$m$ distribution is versatile by adjusting parameter $m$. For instance, when $m$ equals 0.5, it resembles the half-normal distribution, and with $m = 1$, it corresponds to the Rayleigh distribution\cite{Hashemi1993}. Moreover, when $m$ is equal to $(K+1)^2 (2K+1)^{-1}$, the Nakagami-$m$ distribution can closely approximate the Rician distribution with factor $K$\cite{Yacoub1999}.}.
In industrial scenarios, multiple-input multiple-output (MIMO) technology is often used to enhance reliability. As the on-site measurements of composite channel gains show perfect fit to the gamma distribution\cite{Willhammar2024}, this analysis extends to MIMO systems as well.

Denoting the transmission power of UE $U_i$ as $P_i$, the instantaneous SINR at the serving BS for the UE is:

\begin{equation}
      \gamma = \frac{P_0 g_{0} L_{00}}{\mathcal{N} + \mathcal{I}},
      \label{eq_SINR}
\end{equation}
where $g_0$ represents the small-scale channel gain between $U_0$ and $C_0$, which follows a gamma distribution with parameters $m_0$ and $m_0^{-1}$ ($g_0 \sim \Gamma(m_0, m_0^{-1})$). $\mathcal{N}$ denotes the instantaneous power of zero-mean AWGN with variance $\sigma^2$.

$\mathcal{I}$ represents the received interfering power:
\begin{equation}
      \mathcal{I} = \sum_{i \in \mathbf{U}} P_i g_i L_{i0},
      \label{eq_rxIntf}
\end{equation}
with $g_i \sim \Gamma(m_i, m_i^{-1})$ representing the small-scale channel gain between $U_i$ and $C_0$. $\mathcal{S} = P_0 g_{0} L_{00}$ is defined as the instantaneous received power of the signal link.

\section{Deterministic Communication Capacity (DCC)}\label{DCC}

The trade-offs among capacity, reliability and latency in deterministic wireless communication systems are tricky. As latency requirements can be tackled through many ways like shortening transmission intervals (exemplified by mini-slots) and utilizing less resources (for instance, increasing capacity), in this paper, we primarily focus on the capacity and the reliability requirements.

The formula proposed in \cite{Polyanskiy2010} provided valuable theoretical tool for evaluating \emph{instantaneous} capacity. However, the implication of instantly known SINR is hardly feasible in the real-world wireless networks, especially with interference. As a result, with imperfect and delayed SINR information, we should take the additional uncertainty into consideration, leading to reduced \emph{achievable} rate, defined as DCC. In this section, we will introduce a theoretical framework for analyzing the DCC.

\subsection{Instantaneous Capacity and BLER in the FBL Regime}
In the FBL regime, the instantaneous capacity, measured by bits per channel use (bpcu), is approximated as \cite{Polyanskiy2010}
\begin{equation}
      R \approx \log_2(1+\gamma) - \sqrt{\frac{V(\gamma)}{n}} Q^{-1}(\varepsilon),
      \label{eq_InstRate}
\end{equation}
where $V(\gamma)$ is the channel dispersion, $n$ is the number of channel uses,
and $Q^{-1}(x)$ is the inverse Gaussian Q-function.
Thus, the instantaneous BLER is
\begin{equation}
      \varepsilon \approx Q\left(\sqrt{\frac{n}{V(\gamma)}}\left[\log_2(1+\gamma) - R\right]\right),
      \label{eq_InstBLER}
\end{equation}
where $Q(x)$ is the Gaussian Q-function.

Note that, the dispersion of AWGN channel is \cite{Polyanskiy2010}
\begin{equation}
      V_{AWGN}(\gamma) = (\log_2 e)^2 \left[1-(1+\gamma)^{-2}\right].
      \label{eq_AWGNdispersion}
\end{equation}
However, in fading channel, treating the interference as noise will likely make the composite noise\footnote{Conversely, we can also name it as equivalent interference (EI).} deviated from Gaussian distribution, rendering \eqref{eq_AWGNdispersion} unusable. In this case, under the mild assumption of interference independence, we will use the dispersion derived for non-Gaussian noise (NGN)\cite{Scarlett2017} instead
\begin{equation}
      V_{NGN}(\xi;\gamma) = (\log_2 e)^2 \frac{(\xi-1)\gamma^2+4\gamma}{2(\gamma+1)^2},
      \label{eq_NGNdispersion}
\end{equation}
where $\xi$ is the fourth moment of the composite noise. Since for Gaussian noise, $\xi=3$, one can directly find that $V_{NGN}(3;\gamma) = V_{AWGN}(\gamma)$. Therefore, in this paper, we will use \eqref{eq_NGNdispersion} for channel dispersion $V(\gamma)=V_{NGN}(\xi;\gamma)$, with $\xi$ specified according to the distribution of the composite noise.

\subsection{DCC Analysis Framework}
The basic structure of the framework is shown in Fig.~\ref{fig_perfFw}, consisting of 6 modules:
\begin{figure}[!t]
      \centering
      \includegraphics[width=3in]{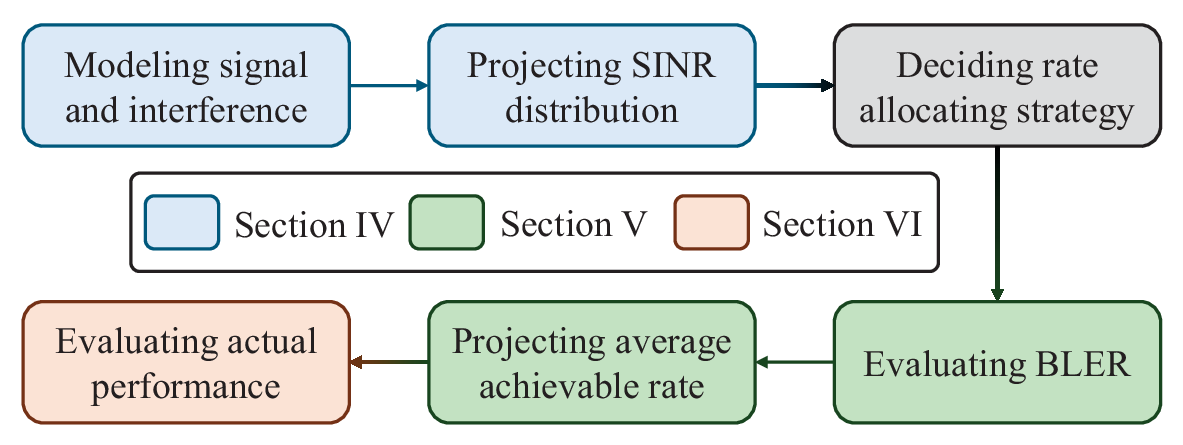}
      \caption{Framework for analyzing DCC in industrial wireless networks.}
      \label{fig_perfFw}
\end{figure}

\subsubsection{Modeling signal and interference, and projecting SINR distribution}
The uncertainty brought by imperfect and delayed SINR information results in the gap between DCC and the instantaneous capacity. Considering different capability of signal and interference cognition, we will statistically model their precision into different levels separately. Then, combining them yields different cases of SINR uncertainty, and the distribution of SINR can be projected. Further deductions of these two modules will be elaborated in Section \ref{statModel}.

\subsubsection{Deciding rate allocating strategy}
Before proceeding to the performance analysis, the scheduling strategy should be determined, based on the precision of the signal. When instantaneous signal power acquisition is possible, the allocated rate can thus be \emph{instantly and accordingly adjusted}. While if only the statistical properties of the signal power can be obtained, the transmission is implied with \emph{fixed rate allocation}.

\subsubsection{Evaluating BLER}
The instantaneous BLER can be evaluated with \eqref{eq_InstBLER}. And taking the expectation of the instantaneous BLER with respect to the SINR yields the average BLER:
\begin{equation}
      \mathbb{E}_\gamma \left[\varepsilon\right] = \int_{0}^{\infty} f_{\gamma}(x) Q\left(\sqrt{\frac{n}{V(x)}}\left[\log_2(1+x) - R\right]\right) \mathrm{d}x.
      \label{eq_expBLER}
\end{equation}

\subsubsection{Projecting average DCC}\label{rateCal}
If both signal and interference are instantly available, then taking the expectation of the instantaneous rate with respect to the instantaneous SINR yields the average DCC:
\begin{equation}
      \mathbb{E}_\gamma \left[R\right] = \mathbb{E}_\gamma \left[\log_2 (1+\gamma)\right] - \frac{Q^{-1}(\varepsilon_\mathrm{req})}{\sqrt{n}} \mathbb{E}_\gamma \left[\sqrt{V(\gamma)}\right].
      \label{eq_expRateInst}
\end{equation}

However, if only their statistical properties are available, \eqref{eq_expRateInst} will be unusable, due to the fact that there is no way to guarantee the instantaneous BLER is constant at the requirement, i.e.,$\varepsilon_\mathrm{req}$ in this case. Thus, the average DCC should make sure the expectation of the BLER meets $\varepsilon_\mathrm{req}$. That is, solving $R$ to satisfy the formula:
\begin{equation}
      \varepsilon_\mathrm{req} = \int_{0}^{\infty} f_{\gamma}(x) Q\left(\sqrt{\frac{n}{V(x)}}\left[\log_2(1+x) - R\right]\right) \mathrm{d}x.
      \label{eq_expRateMean}
\end{equation}
Actually, as the rate allocation is fixed, $R$ is both the instantaneous and the average achievable rate. Due to the extreme complexity of analytically solving \eqref{eq_expRateMean} for $R$, in this paper, we majorly conduct asymptotic analysis for these cases. 

There is also a tricky situation that the signal power is instantly available, but only the statistical properties of the interference power is known. In that situation, signal and interference should be separately considered:
\begin{equation}
      \begin{aligned}
            &\mathbb{E}_{\mathcal{S},\mathcal{I}} \left[R(\mathcal{S},\mathcal{I})\right]\\
            =&\mathbb{E}_{\mathcal{S}} \{\mathbb{E}_{\mathcal{I}}\left[R(\mathcal{S},\mathcal{I})\right]\}
      \end{aligned}
\end{equation}

For the inner expectation, similar to the situation of \eqref{eq_expRateMean}, $R$ can be solved with respect to the interference, thus reduced to $R(\mathcal{S})$. Next, solving the outer expectation will yield the desired result.

\subsubsection{Evaluating actual performance}
Finally, the evaluation result will be presented in Section \ref{Sim}, and it would be calculated according to the actual distribution of SINR, the allocated rate, and related parameters.

\subsubsection{Differences from existing approaches}
Existing frameworks for analyzing industrial wireless networks fall into two categories: stochastic geometry models based on statistical device distributions, and deterministic models assuming perfect SINR knowledge. Both suffer critical practical limitations.

Stochastic geometry methods\cite{Zhang2024, Hesham2024, Hesham2024a} offer general network insights but lack spatial specificity: their agnostic approach to deployment configurations makes them unsuitable for evaluating concrete scenarios. Conversely, perfect-SINR assumptions\cite{Tran2021, Hoang2024,Gu2018} enable instantaneous performance analysis of specific networks, yet this is fundamentally unworkable in practice, especially for grant-free systems where real-time SINR feedback is infeasible. The inherent variability of interference power (as elaborated later) further undermines their applicability.

The proposed DCC framework addresses these limitations through two key innovations: first, shifting analytical focus from device distributions to SINR statistics, enabling precise performance evaluation without requiring instant feedback; second, adopting average DCC as the primary metric instead of instantaneous measures. This approach is particularly critical for grant-free systems where average DCC becomes the sole viable basis for rate allocation. In grant-based networks, it provides a stable long-term perspective that complements instantaneous measurements, offering more comprehensive network efficiency insights.

\section{Interference Cognition and SINR Distribution}
\label{statModel}
Existing works, like \cite{Zhang2022}, have demonstrated the importance of SINR prediction for analyzing system performance and guaranteeing deterministic transmission even with coarse interference cognition. However, as evident from \eqref{eq_SINR} and \eqref{eq_rxIntf}, interference cognition plays a crucial role in SINR prediction. This section analyzes the capability of SINR prediction for different levels of interference cognition precision.

\subsection{Precision of Interference Cognition}
The received interference arises from multiple contributing links, and the identifiability of these links is crucial for modeling. Based on this and the precision of power estimation, we categorize interference cognition precision into four levels:
\begin{itemize}
      \item Level I: The instantaneous power of each link can be accurately obtained.
      \item Level D: The distribution of power for each link can be determined.
      \item Level A: The average power of each link can be determined.
      \item Level M: Only the distribution of the aggregated received interference can be inferred.
\end{itemize}

Levels I, D, and A provide link-level interference cognition. Level I represents the ideal case, which may be challenging to achieve in real-world systems and serves primarily as a theoretical upper bound. However, reaching the precision of Level A is feasible with techniques like NLRA \cite{Peng2022}. 

Conversely, Level M does not offer per-link interference information.

Figure \ref{fig_lvl_intf} illustrates these four levels using two orange and blue interference links as examples:
\begin{figure}[!t]
      \centering
      \includegraphics[width=3in]{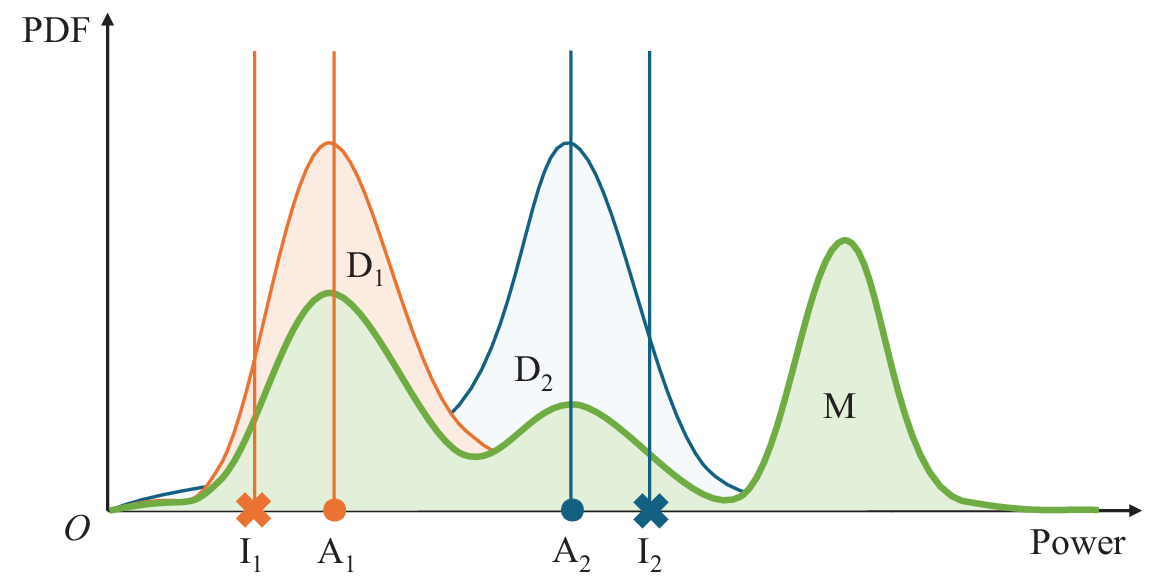}
      \caption{The four levels of precision of the interference.}
      \label{fig_lvl_intf}
\end{figure}

For Level I, the exact power for each link is known, represented by the two "$\times$" markers. For Level D, the PDF of the power is known (due to its known distribution), shown as orange and blue thin solid lines in the figure. While only the average powers of the links are known in Level A, signified by the "$\bullet$" markers.

Finally, the green solid bold line represents the probability density function (PDF) obtainable in Level M, showing peaks corresponding to every possible case with a given number of active interfering links (in Fig.~\ref{fig_lvl_intf}, the number is one).

For the signal part, with OFDMA, each UE and its associated BS have only one signal link. Therefore, we consider the following two levels:
\begin{itemize}
      \item Level I: The instantaneous power of the link can be accurately obtained.
      \item Level D: The power distribution of the link can be determined.
\end{itemize}

\subsection{SINR Distribution Projection with Different Interference Cognition}
The precision of cognition determines the uncertainty in the information used for SINR estimation, leading to different projected SINR distributions. We discuss the distributions of interference and signal separately before combining them for the SINR distribution.
\subsubsection{Distributions of instantaneous interference power}
\begin{itemize}
      \item Level I: As the instantaneous power of each interfering link is known, the interference is deterministic.
      \item Level D: Each per-link interference $I_i$ follows a Nakagami-$m$ distribution with shape parameter $m_i$ and spread parameter $\Omega_i = P_i L_{i0}$. The noise magnitude $|N|$ also follows a Nakagami-$m$ distribution with $m_N = 0.5$ and $\Omega_N = \sigma^2$. As a result, the EI power (EIP) $\mathcal{I}_E$ satisfies
      $$\mathcal{I}_E = \mathcal{I} + \mathcal{N} = \sum_{i \in \mathbf{U}} I_i^2 + |N|^2 \sim \Gamma(m_I, \frac{\Omega_I}{m_I}),$$
      with \cite{NAKAGAMI1960}
      \begin{equation}
            m_I \approx \frac{\left(\Omega_N + \sum_{i=1}^{|\mathbf{U}|} \Omega_i\right)^2}{2\Omega_N^2 + \sum_{i=1}^{|\mathbf{U}|} \left(\frac{\Omega_i^2}{m_i}\right)},
            \label{eq_aggm}
      \end{equation}
      and
      \begin{equation}
            \Omega_I = \Omega_N + \sum_{i=1}^{|\mathbf{U}|} \Omega_i.
      \end{equation}
      Thus, we have $\xi=1+m_I^{-1}$ for projections based on level D and I, and also in reality.
      \item Level A: Since the distribution is unknown, we treat it as a Gaussian distribution ($m_I = 0.5$). Similar to Level D, we have $\mathcal{I}_E \sim \Gamma(0.5, 2\Omega_I)$, and $\xi=3$ for projections.
      \item Level M: The distribution of aggregated received interference can be inferred from historical data. Other than that, only the number of scheduled UEs can be obtained. Thus, every combination of UEs of the same number can occur. As each combination corresponds to the active UE set of a level D case, the level M case can be seen as combining multiple level D cases with different sets of active interferences. Therefore, a $K$-component mixture of gamma distribution is adopted:
      \begin{equation}
            f_{\mathcal{I}_E}(x) = \sum_{k=1}^{K} \frac{f_{\mathcal{I}_{E,k}}(m_{I,k},\Omega_{I,k};x)}{K},
            \label{eq_pdfIM}
      \end{equation}
      where $K = \binom{|\mathbf{A}|}{|\mathbf{U}|}$, and $\mathcal{I}_{E,k} \sim \Gamma(m_{I,k}, \frac{\Omega_{I,k}}{m_{I,k}})$, with $m_{I,k}$ and $\Omega_{I,k}$ being the shape and spread parameter of the $k$-th combination of active interference sources. Thus, the PDF of $\mathcal{I}_{E,k}$ is $$f_{\mathcal{I}_{E,k}}(m_{I,k},\Omega_{I,k};x) = \frac{1}{\Gamma(m_{I,k}) \Omega_{I,k}^{m_{I,k}}} x^{m_{I,k}-1} e^{-x/\Omega_{I,k}}.$$
\end{itemize}

\subsubsection{Distributions of instantaneous received signal power}
\begin{itemize}
      \item Level I: Similarly, the instantaneously obtained signal power is deterministic.
      \item Level D: The distribution of signal also follows a Nakagami-$m$ distribution with shape parameter $m_0$ and spread parameter $\Omega_0 = P_0 L_{00}$. Thus, we have $\mathcal{S} \sim \Gamma(m_0, \frac{\Omega_0}{m_0})$.
\end{itemize}

\subsubsection{Distributions of projected SINR}
Combining different precision levels of interference and signal yields 8 cases of cognition precision in total.
\footnote{To streamline the discussion, we use a simplified notation, with each case represented as "signal power precision/EIP precision". For example, if the instant power of the signal link can be collected accurately (Level I), and the per-link average power of the interference can be determined (level A), it will be denoted as "I/A".}
However, the D/I case is not feasible, due to the unlikelihood of obtaining the power of interference links with higher precision than that of the signal link. Thus, there are 7 cases possible, and the corresponding projected SINR distributions are listed in Table \ref{tbl_SinrDist}.
\begin{table}[!t]
      \caption{Distribution of Projected SINR for Different Interference Cognition}
      \begin{center}
            \label{tbl_SinrDist}
            \begin{tabular}{|c|c|c|}
            \hline
            \textbf{Signal}   &\textbf{Interference}  &\textbf{Distribution of SINR} \\
            \hline
            \multirow{3}{*}{I}      &I                &Deterministic\\
                                    \cline{2-3}
                                    &D\&A             &Inverse-gamma distribution\tablefootnote{By definition, the reciprocal of a Gamma-distributed (with parameters $k$ and $\theta$) random variable (RV) yields the inverse-gamma-distributed RV with parameters $k$ and $\theta^{-1}$.}\\
                                    \cline{2-3}
                                    &M                &Mixture of Inverse-gamma distribution\\
            \hline
            \multirow{2}{*}{D}      &D\&A             &F-distribution\tablefootnote{Precisely speaking, it is the "regularized" SINR, $\hat{\gamma} = \gamma \Omega_I/\Omega_0$, that subjects to the F-distribution with parameters $2m_0$ and $2m_I$.}\\
                                    \cline{2-3}
                                    &M                &Mixture of F-distribution\tablefootnote{Similar as above, technically, "regularized" SINR of each component is subject to the F-distribution.}\\
            \hline
            \end{tabular}
      \end{center}
\end{table}

\section{DCC Analysis}\label{perfAnal}
Unlike many existing works that primarily focus on BLER, this study emphasizes the average DCC and its asymptotic behavior. Understanding these aspects is crucial, especially for resource allocation and rate analysis in industrial wireless communication systems.

\subsection{Analytical Analysis of the Average DCC}
The analytical expression of the average DCC is intractable for general cases. However, as in the I/I case, we can actually benefit from highly fluctuating interference power. Thus, its upper bound can be found by setting $m_I$ to the minimum, i.e. $m_I = 0.5$, which yields Theorem \ref{th_avgRate}.

\begin{theorem}\label{th_avgRate}
      The upper bound of average DCC in the I/I case, with $m_I = 0.5$, is
      \begin{equation}
            \bar{R} = \mathcal{C}_1 - \frac{Q^{-1}(\varepsilon_\mathrm{req})}{\sqrt{n}} \mathcal{C}_2,
            \label{eq_R1}
      \end{equation}
      where $\mathcal{C}_1$ is defined in \eqref{eq_C1} with $m_I=0.5$, and $\mathcal{C}_2$ is defined in \eqref{eq_C2}.
\end{theorem}
\begin{proof}
      See Appendix \ref{proof_th_avgRate}.
\end{proof}
Note that, from the derivation, \eqref{eq_C1} is independent of the channel dispersion, thus can be applied to other cases.

Unfortunately, deriving the average DCC analytically for other cases is hardly possible. However, we can analyze their asymptotic behavior, which is discussed in the next subsection.

\subsection{Asymptotic Analysis of the Average DCC}\label{avgAsy}
While analytical solutions are challenging for some cases, analyzing their asymptotes reveals important characteristics of the average DCC. This subsection presents the asymptotic expressions for five different cases under varying SINR precision levels and analyzes their implications.

\subsubsection{I/I}
The asymptotic expression for the average DCC in the I/I case can be derived from Theorem \ref{th_avgRate}.
\begin{theorem}\label{th_asyII}
      The asymptote of the average DCC in the I/I case is
      \begin{equation}
            \bar{R}_{I/I} = \frac{\log_2 10}{10}\bar{\gamma}\mathrm{(dB)} + \mathcal{C}_3,
            \label{eq_asyII}
      \end{equation}
      where $\bar{\gamma}\mathrm{(dB)} = 10 \log_{10} \bar{\gamma}$, $\bar{\gamma}$ is the average SINR, and $\mathcal{C}_3$ is defined in \eqref{eqa_IIbfinal}.
\end{theorem}
\begin{proof}
      See Appendix \ref{proof_th_asyII}.
\end{proof}

As evident from \eqref{eq_asyII}, the asymptotic performance is influenced by shape parameters of both signal and interference distributions. Whereas the rest terms exhibit similarities to \eqref{eq_InstRate}, reflecting the intuitive relationship: since instantaneous DCC adapts to real-time SINR values according to \eqref{eq_InstRate}, its average counterpart should share analogous characteristics.

\subsubsection{I/D and I/A}
In these cases, only the signal power is known instantaneously. Therefore, signal and interference require separate treatment as described in Section \ref{rateCal}. 

First, rewriting \eqref{eq_expRateMean} for these cases
\begin{equation}
      \begin{aligned}
            &\varepsilon_\mathrm{req} = \\
            &\int_{0}^{\infty} f_\gamma\left(\mathcal{S};x\right) Q\left(\sqrt{\frac{n}{V(x)}}\left[\log_2(1+x) - R\left(\mathcal{S};x\right)\right]\right) \mathrm{d}x.
      \end{aligned}
      \label{eq_expRateMeanID}
\end{equation}

Due to the complexity of $R\left(\mathcal{S};x\right)$, we approximate it using its asymptote. At any given moment, the signal power is known and can be treated as a constant value, while the distribution of per-link interference can be determined (Level D) in this case. Given that the statistical properties of both SINR components are not subject to instantaneous fluctuations, this scenario inherently implies the use of a fixed-rate allocation strategy.
\begin{lemma}\label{le_asyCD}
      The asymptote of average DCC when the signal power is a known constant, and the interference information obtained reaches Level D, is:
      \begin{equation}
            \bar{R} = \frac{\log_2 10}{10}\bar{\gamma}\mathrm{(dB)} + \mathcal{C}_4.
            \label{eq_asyCD}
      \end{equation}
      Where, if $m_I \ne 1$, $\mathcal{C}_4$ is defined in \eqref{eq_C3}. Otherwise, it is defined in \eqref{eq_C3-1}.
\end{lemma}
\begin{proof}
      See Appendix \ref{proof_le_asyCD}.
\end{proof}

In Appendix \ref{proof_th_asyID_k}, we prove that the slope of the asymptote for the SINR(dB)-to-average-DCC curve is independent of the SINR distribution when using fixed rate allocation. This leads to the following corollary:

\begin{corollary}\label{cor_slopeInvar}
      With fixed rate allocation, the slope of the asymptote of the SINR(dB)-to-average-DCC curve is constant at $\frac{\log_2 10}{10}$, independent of the distribution of the SINR.
\end{corollary}

We can now derive the asymptote for the I/D case:
\begin{theorem}\label{th_asyID}
      The asymptote of the average DCC in the I/D case is
      \begin{equation}
            \bar{R}_{I/D} = \frac{\log_2 10}{10}\bar{\gamma}\mathrm{(dB)} + \psi(m_0)\log_2e - \log_2 m_0 + \mathcal{C}_4.
            \label{eq_asyID}
      \end{equation}
\end{theorem}
\begin{proof}
      Rewriting \eqref{eq_asyCD} from Lemma \ref{le_asyCD}:
      \begin{equation}
            R = \log_2\frac{\mathcal{S}}{\bar{\mathcal{I}_E}} + \mathcal{C}_4,
            \label{eq_IDs}
      \end{equation}
      where $\bar{\mathcal{I}}_E$ is the average EIP.
      Taking the expectation of $R$ with respect to $\mathcal{S}$ yields
      \begin{equation}
            \begin{aligned}
                  \mathbb{E}_\mathcal{S}\left[R\right]
                  &\approx \mathbb{E}_\mathcal{S}\left[\log_2\mathcal{S}\right] - \log_2\bar{\mathcal{I}_E} + \mathcal{C}_4 \\
                  &= \log_2e\left[\psi(m_0)+\log \bar{\mathcal{S}}-\log m_0\right] - \log_2\bar{\mathcal{I}_E} + \mathcal{C}_4 \\
                  &= \log_2 \bar{\gamma} +\psi(m_0)\log_2e - \log_2 m_0 + \mathcal{C}_4,
            \end{aligned}
      \end{equation}
      where $\psi(x)$ is the digamma function\cite[5.2.2]{DLMF}, and $\bar{\mathcal{S}}$ is the average signal power.
      Adopting the decibel notation gives us \eqref{eq_asyID}.
\end{proof}
Similarly, setting $m_I = 0.5$ in Theorem \ref{th_asyID} allows for deriving the asymptote in the I/A case.

\subsubsection{D/D and D/A}\label{asyDD}
These two scenarios are commonly analyzed in grant-free transmission frameworks, where instantaneous feedback mechanisms are inherently unfeasible. Under such constraints, fixed-rate allocation emerges as the sole technically viable solution. Consequently, all subsequent derivations are based on \eqref{eq_expRateMean}.
\begin{theorem}\label{th_asyDD}
      The asymptote of the average DCC in the D/D case is:
      \begin{equation}
            \bar{R}_{D/D} = \frac{\log_2 10}{10}\bar{\gamma}\mathrm{(dB)} + \mathcal{C}_5,
            \label{eq_asyDD}
      \end{equation}
      where $\mathcal{C}_5$ is derived in \eqref{eq_C5}.
\end{theorem}
\begin{proof}
      See Appendix \ref{proof_th_asyDD}.
\end{proof}

From Theorem \ref{th_asyDD}, we can establish a relationship between $\varepsilon_\mathrm{req}$ and the intercept of the asymptote
\begin{corollary}\label{cor_epsForFixed}
      If fixed rate allocation is adopted, the intercept of the asymptote is linear to the logarithm of the BLER requirement
      \begin{equation}
            \mathcal{C}_5 = {m_0}^{-1} \log_2\varepsilon_\mathrm{req} + (\text{terms invariable to }\varepsilon_\mathrm{req}).
            \label{eq_C5vsEPS}
      \end{equation}
\end{corollary}
\begin{proof}
      Rearranging \eqref{eq_C5} directly yields \eqref{eq_C5vsEPS}.
\end{proof}

The linear relationship presented in Corollary \ref{cor_epsForFixed} suggests that using multiple repetitions might be profitable for mitigating the DCC penalty by relaxing the BLER target for each replicated transmission.

Similarly, the asymptote of the D/A case can be derived from Theorem \ref{th_asyDD} by setting $m_I = 0.5$.

\section{Numerical Results and Analysis}\label{Sim}
This section evaluates the proposed theorems analyzes their impact on average DCC, considering the precision of obtained power information.

\subsection{Simulation Configurations}
The simulations presented in this section serve to demonstrate the practical applicability of our proposed framework and associated theorems. It is important to emphasize that, as the theoretical derivations are fundamentally based on SINR analysis, they maintain inherent environment-independence: a characteristic stemming from the mathematical formulation rather than empirical assumptions. While we fully recognize the heterogeneous nature of manufacturing systems and the vast diversity of operational scenarios encountered in industry, it is both impractical and unnecessary to simulate every possible configuration within this section.

To address this challenge, our approach employs two complementary strategies: Section \ref{simRes} presents comprehensive simulations using a representative parameter set that captures typical system behavior. Meanwhile, the rest of Section \ref{allRes} provides detailed analyses of individual parameter impacts through targeted simulation studies, ensuring both broad applicability and nuanced understanding of the framework's behavior under varying conditions.

Inside the manufacturing plant, there are $\lambda_r = 4$ rows, with $\lambda_c = 3$ columns per row of cells. Each cell is a squared space with a side length of $r = 50$ m and adjacent to each other. One BS and multiple UEs are randomly placed within each cell, and UEs connect to the BS inside the same cell. Fig.~\ref{fig_gen_sim} illustrates one such scenario. Each UE takes the semi-persistence criterion ($-67$ dBm) as the baseline power $P$. Denoting the distance between $U_i$ and $C_{z_j}$ as $d_{ij}$, the pathloss model used in simulation is
\begin{equation}
      L_{ij} = 38.46 + 20 \log_{10}{d_{ij}}.
\end{equation}
Further, partial path loss compensation is applied with a factor $\delta \in (0,1)$, resulting in $P_i = P\times L_{ii}^{\delta}$. Besides, an upper bound of transmission power $P_{max} = 23$ dBm is applied to all the UEs.

\begin{figure}[!t]
      \centering
      \includegraphics[width = 0.35\textwidth]{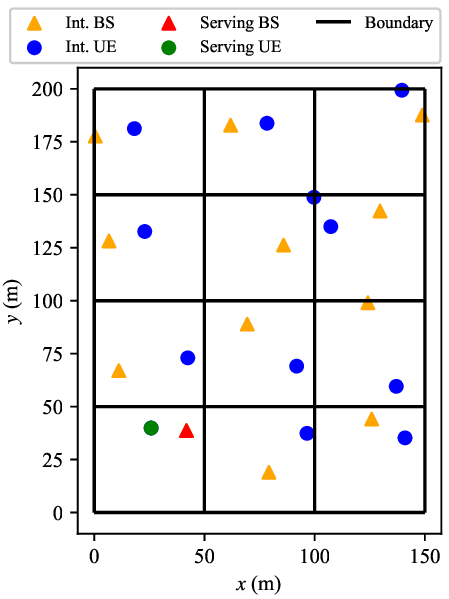}
      \caption{One of the generated scenarios.}
\label{fig_gen_sim}
\end{figure}

Unless otherwise specified, results presented in this section are based on the parameters listed in Table \ref{tbl_SimParams}.

\begin{table}[!t]
      \caption{Simulation Parameters}
      \begin{center}
            \label{tbl_SimParams}
            \begin{tabular}{|c|c|}
            \hline
            \textbf{Parameter}&\textbf{Value} \\
            \hline
            Side length $r$& 50 m \\
            Number of rows $\lambda_r$& 4 \\
            Number of columns $\lambda_c$& 3 \\
            Path loss compensation factor $\delta$& 0.7 \\
            Path loss model& $(38.46 + 20\log_{10}{d})$ dB\\
            White noise power density& -174 dBm/Hz \\
            Channel uses $n$& 200 \\
            Baseline transmission power $P$& -67 dBm \\
            Maximum power $P_{max}$& 23 dBm \\
            Shape parameter $m_0$ and $m_I$& 2\\
            BLER requirement $\varepsilon_\mathrm{req}$& $10^{-5}$\\
            \hline
            \end{tabular}
      \end{center}
\end{table}

\subsection{Evaluation of Theorem \ref{th_avgRate}}
We compare the proposed theorem with the result obtained through numerical integration to assess its accuracy. Fig.~\ref{fig_th1}(a) demonstrates high similarity between the derived result and simulated values, supporting the validity of our proposed theorem.

\begin{figure}[!t]
      \centering
      \begin{minipage}[b]{0.48\textwidth}
      \centering
      \subfigure[]{\includegraphics[width = 0.48\textwidth]{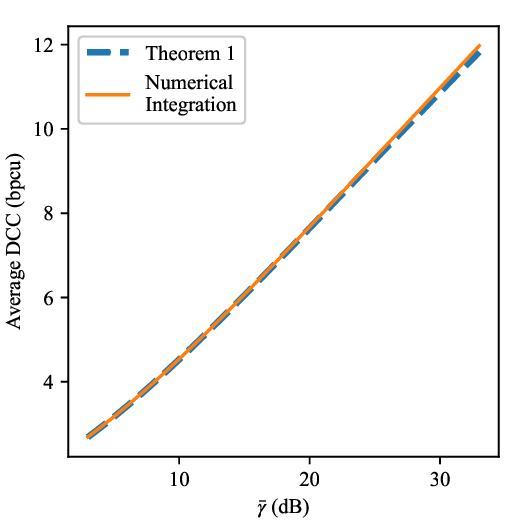}}
      \subfigure[]{\includegraphics[width = 0.48\textwidth]{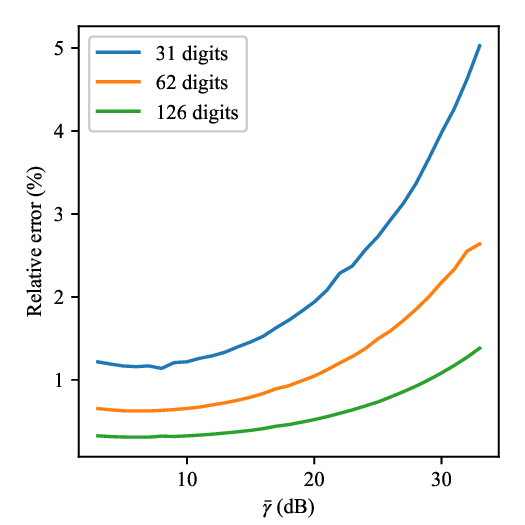}}
      \end{minipage}
      \caption{Evaluation of Theorem \ref{th_avgRate} with $m_I=0.5$.
               (a) Exact value of the average DCC, calculated by numerical integration and Theorem \ref{th_avgRate}.
               (b) The relative error between the value of numerical integration and Theorem \ref{th_avgRate} with 31-, 62-, and 126-digit precision.}
\label{fig_th1}
\end{figure}
Furthermore, the relative error (Fig. \ref{fig_th1}(b)) illustrates that these measures are indeed minimal, further confirming the accuracy of our derived result. It is important to note that at higher SINR values, numerical precision limitations can be a primary source of error. As numerical precision improves, the difference between theoretical and simulated results approaches near-zero levels.

\subsection{Evaluation of the asymptotic analysis and the impact of different factors}\label{allRes}
This subsection analyzes how the accuracy of power information and parameter configurations affect average DCC. We will also explore the implications of varying degrees of precision in power data.

\subsubsection{The impact of power information precision}\label{simRes}
Fig.~\ref{fig_asy_difscenario} demonstrates the rate and BLER performance under different power information precision scenarios.

\begin{figure}[!t]
      \centering
      \begin{minipage}[b]{0.48\textwidth}
      \centering
      \subfigure[]{\includegraphics[width = 0.48\textwidth]{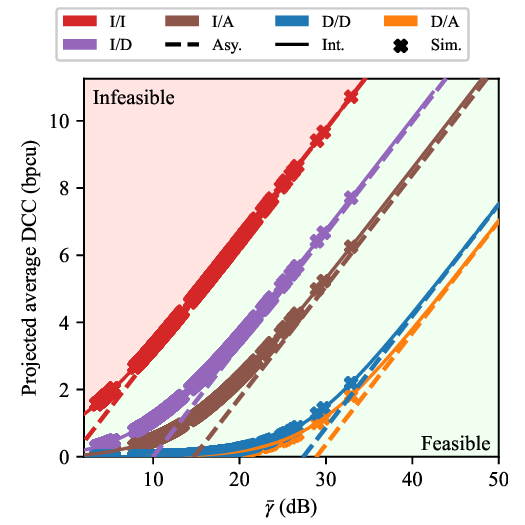}}
      \subfigure[]{\includegraphics[width = 0.48\textwidth]{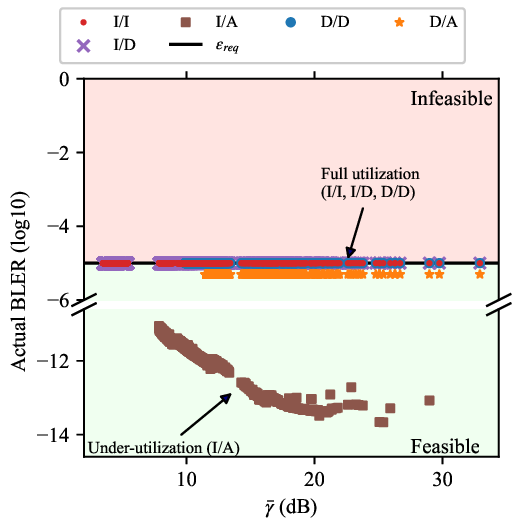}}
      \end{minipage}
      \caption{Performance under cases with per-link power information precision.
               (a) Projected average DCC. (b) Actual BLER.
               Remark: different colors indicate different scenarios. Besides, in (a), different line styles and marker represent different data sources. The dashed lines, solid lines, and "$\times$" markers correspond to asymptotes (Asy.), numerical integration results (Int.), and simulated results (Sim.), respectively. In (b), markers are used in combination with colors to indicate different scenarios, and the solid line indicates the reliability requirement ($\varepsilon_\mathrm{req}$).}
\label{fig_asy_difscenario}
\end{figure}
Fig.~\ref{fig_asy_difscenario}(a) reveals two distinct clusters of curves. The precision of signal power significantly impacts the projected average achievable rate. Interestingly, under the same signal power precision, the average DCC-SINR curves for cases with per-link interference power available (I, D, and A) are closely grouped. This observation highlights the importance of obtaining per-link interference power information. While further increasing the precision of interference power marginally improves the average DCC, level A should be the preferred choice due to its balance between complexity and performance.

Combining Fig.~\ref{fig_asy_difscenario}(b), we infer that reducing signal and interference power information precision from level I to level D negatively affects actual DCC.  For example, the I/I and I/D cases exhibit a BLER just meeting the target (indicating projected DCC is close to actual DCC), but their projected average DCC differs, with I/I achieving higher DCC than I/D. The same trend holds between the I/D and D/D cases.

Conversely, the difference in projected DCC between levels D and A of interference power information is not due to reduced \emph{actual} DCC but from insufficient channel utilization (evident by a significantly lower BLER, especially for I/A). This suggests that while precise knowledge of signal power is beneficial, it might be less crucial than obtaining per-link interference power.

Furthermore, Fig.~\ref{fig_asy_difscenario}(a) shows that the asymptotes derived in Section \ref{avgAsy} closely align with numerical integration curves for the same case, justifying their use for accurate and easy performance evaluation.

\begin{figure}[!t]
      \centering
      \begin{minipage}[b]{0.48\textwidth}
      \centering
      \subfigure[]{\includegraphics[width = 0.48\textwidth]{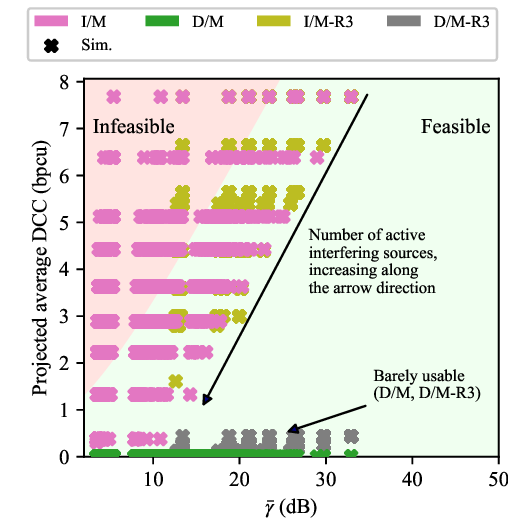}}
      \subfigure[]{\includegraphics[width = 0.48\textwidth]{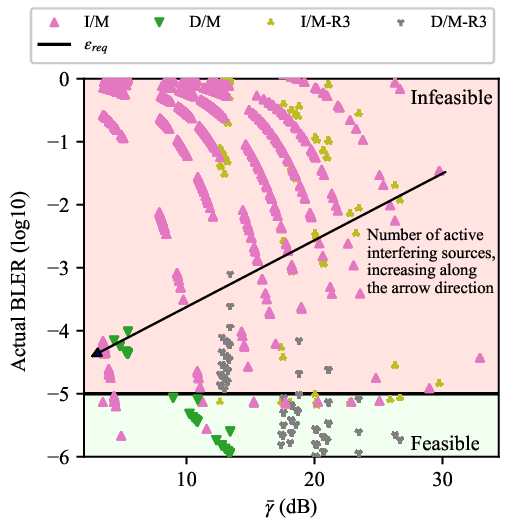}}
      \end{minipage}
      \caption{Performance under cases with aggregated interference information.
               (a) Projected average DCC. (b) Actual BLER.
               Remark: in (a), different colors indicate different scenarios, whereas different marker represent different data sources. The "$\times$" marker in (a) corresponds to simulated results (Sim.), and the solid line in (b) indicates the reliability requirement ($\varepsilon_\mathrm{req}$).}
\label{fig_mcase_difscenario}
\end{figure}
However, when the interference power precision is level M, the projections of average DCC vary significantly with different signal power precision cases, as shown in Fig.~\ref{fig_mcase_difscenario}. For the D/M case, the projected DCC approaches zero, resulting from conservative projections of DCC, which is due to the combination of excessive width of interference PDF \footnote{One can directly infer from \eqref{eq_pdfIM} that, the width of the PDF of the aggregated interference is at least as large as that of its widest component. Moreover, the difference between the largest and smallest spread parameters among the components also contributes to the overall width of the aggregated PDF. As the number of components increases, both of these factors tend to grow, leading to a substantial widening of the aggregated PDF.} and uncertainty of signal link power. Even utilizing the reuse-3 scheme (D/M-R3) can barely improve the projected DCC.

In contrast, for the I/M case, different numbers of interference sources yield different projections. However, for any specified number of interference devices, the projection remains unchanged. As a result, multiple parallel bands can be observed in Fig.~\ref{fig_mcase_difscenario}(a) for I/M case. The same behavior is also present when reuse-3 scheme is adopted (I/M-R3). Compared to the D/M cases, the exact knowledge of signal link power helps in generating viable projections.

Unfortunately, as illustrated in Fig.~\ref{fig_mcase_difscenario}(b), conservative projection of DCC by the D/M case is still unable to fully guarantee the requirement of reliability, especially when SINR is low. As the aggregated interference information is not targeted to any specific combinations of interference devices, utilizing such information to estimate DCC for any specific scenarios can only result in over- or under-estimation, leading to overly high or excessively low BLER. The same mechanism also applies to the I/M case. Thus, significantly higher BLER can be observed as a direct consequence of overly optimistic DCC projections in scenarios involving certain combinations of strong interfering sources or a small number of interferers.

Besides, introducing reuse-3 scheme only reduce the possible interference, but the aggregated nature of the interference information still remains. Thus, the I/M-R3 and the D/M-R3 cases also find it difficult to fully meet the reliability requirement. Conclusively, level M interference power information proves inadequate for industrial wireless networks. Per-link interference power is essential for achieving decent QoS.

\subsubsection{The impact of $m_I$}
\begin{figure}[!t]
      \centering
      \begin{minipage}[b]{0.48\textwidth}
            \centering
            \subfigure[]{\includegraphics[width = 0.48\textwidth]{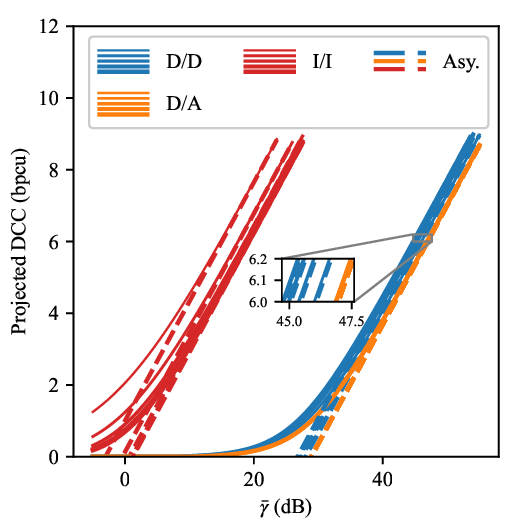}}
            \subfigure[]{\includegraphics[width = 0.48\textwidth]{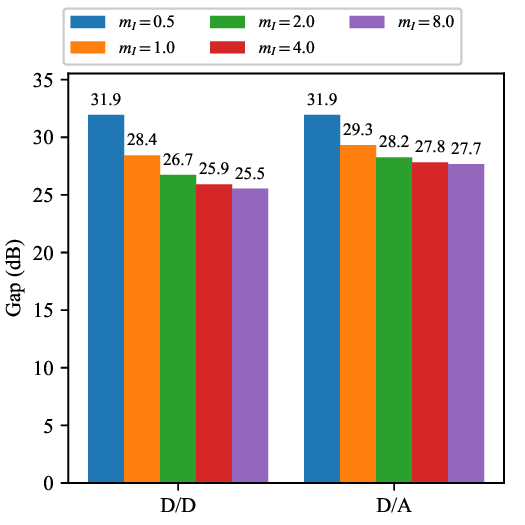}}
      \end{minipage}
      \caption{Performance against different $m_I$ for fixed rate allocation cases.
               (a) Projected average DCC. From the most thin to the most bold, the solid lines correspond to $m_I = 0.5, 1, 2, 4, 8$.
               (b) SINR gap between the ideal case and the cases with fixed rate allocation.}
\label{fig_difm_fixed}
\end{figure}
\begin{figure}[!t]
      \centering
      \begin{minipage}[b]{0.48\textwidth}
            \centering
            \subfigure[]{\includegraphics[width = 0.48\textwidth]{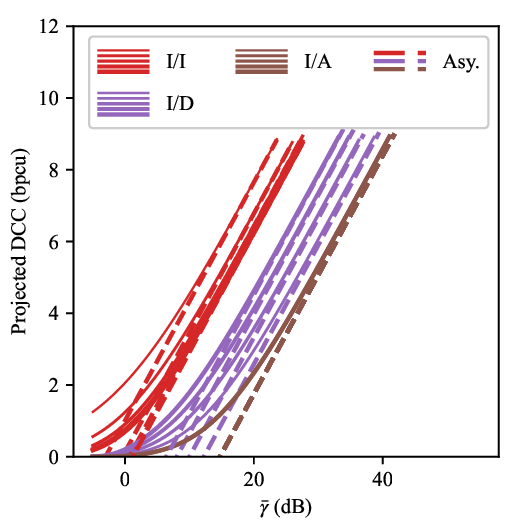}}
            \subfigure[]{\includegraphics[width = 0.48\textwidth]{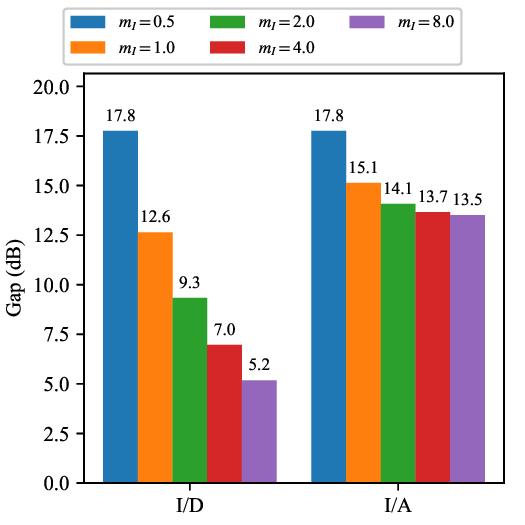}}
      \end{minipage}
      \caption{Performance against different $m_I$ for cases with instantaneous rate adjustment.
               (a) Projected average DCC. From the most thin to the most bold, the solid lines correspond to $m_I = 0.5, 1, 2, 4, 8$.
               (b) SINR gap between the ideal case and the cases with instantaneous rate adjustment.}
\label{fig_difm_insta}
\end{figure}
We discuss the effect on different rate adjustment strategies.
For fixed rate allocation (Fig.~\ref{fig_difm_fixed}), greater $m_I$ leads to slight DCC improvements, but the difference is negligible. Conversely, for cases with instantaneous rate adjustment (Fig.~\ref{fig_difm_insta}), less variable interference links significantly enhance projected DCC, approaching the ideal (I/I) case. Thus, obtaining interference power of level D precision is preferred over level A if possible.

Interestingly, for the ideal case, the trend reverses: as $m_I$ increases, the projected DCC gradually decreases. This occurs because a larger $m_I$ results in a narrower PDF for the interference link power. Consequently, the probability of interference power falling below the average diminishes, while the probability of above-average interference power also decreases. As the DCC is more sensitive to decreases of SINR than increases, their combined effect slightly diminishes the projected DCC.

\subsubsection{The impact of $m_0$}
Fig.~\ref{fig_difm0_fixed} illustrates performance against different $m_0$ for fixed rate allocation cases.
\begin{figure}[!t]
      \centering
      \begin{minipage}[b]{0.48\textwidth}
            \centering
            \subfigure[]{\includegraphics[width = 0.48\textwidth]{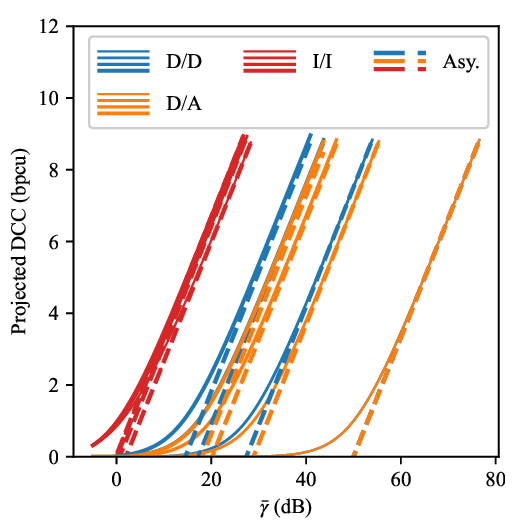}}
            \subfigure[]{\includegraphics[width = 0.48\textwidth]{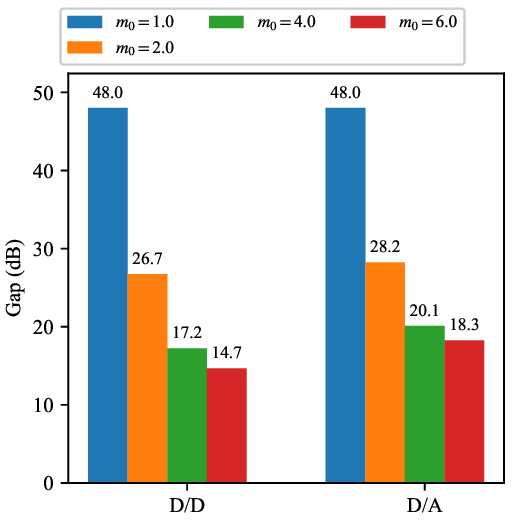}}
      \end{minipage}
      \caption{Performance against different $m_0$ for fixed rate allocation cases.
               (a) Projected average DCC. From the most thin to the most bold, the solid lines correspond to $m_0 = 1, 2, 4, 6$.
               (b) SINR gap between the ideal case and the cases with fixed rate allocation.}
\label{fig_difm0_fixed}
\end{figure}
As seen in Fig.~\ref{fig_difm0_fixed}(a), curves for both D/D and D/A cases rapidly shift leftward with increasing $m_0$, indicating significantly improved performance.

In Section \ref{perfAnal}, we analyze the asymptotic behavior, revealing parallel asymptotes for all cases, allowing us to determine the additional SINR needed to compensate for power information uncertainty. Fig.~\ref{fig_difm0_fixed}(b) shows a 48 dB gap rendering the system unusable for both D/D and D/A cases when $m_0$ is low. As $m_0$ increases, this SINR gap narrows, with the improved power information precision in the D/D case leading to a slightly greater reduction compared to the D/A case. Therefore, for fixed rate allocation, enhancing performance necessitates considering other techniques like multiple repetitions or MIMO, which increase both $m_0$ and effective SINR, ultimately leading to better DCC.

\begin{figure}[!t]
      \centering
      \begin{minipage}[b]{0.48\textwidth}
            \centering
            \subfigure[]{\includegraphics[width = 0.48\textwidth]{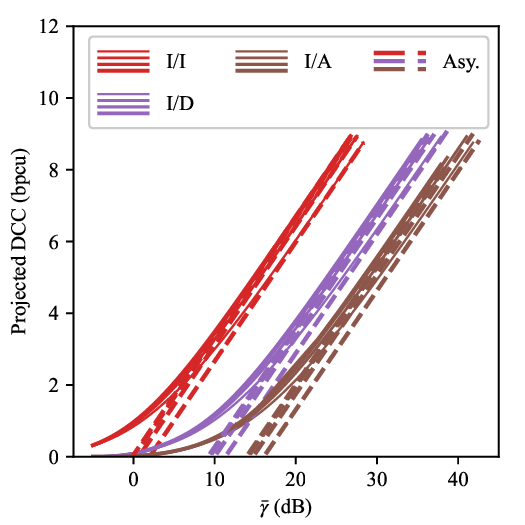}}
            \subfigure[]{\includegraphics[width = 0.48\textwidth]{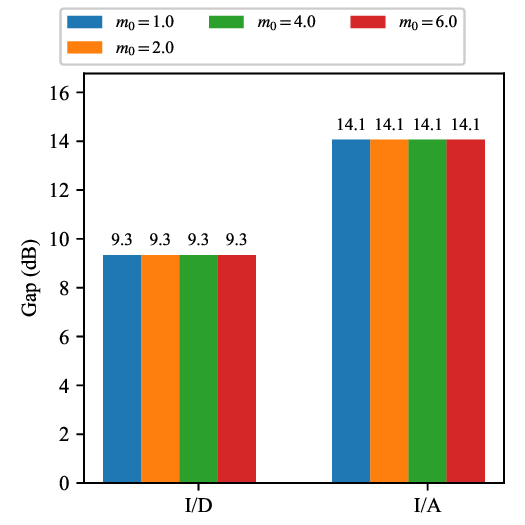}}
      \end{minipage}
      \caption{Performance against different $m_0$ for cases with instantaneous rate adjustment.
               (a) Projected average DCC. From the most thin to the most bold, the solid lines correspond to $m_0 = 1, 2, 4, 6$.
               (b) SINR gap between the ideal case and the cases with instantaneous rate adjustment.}
\label{fig_difm0_insta}
\end{figure}
While improved projected average DCC is observed for all three cases as $m_0$ increases (Fig.~\ref{fig_difm0_insta}(a)), their pace remains nearly identical. Consequently, Fig.~\ref{fig_difm0_insta}(b) shows a constant SINR gap regardless of $m_0$. Therefore, while multiple repetitions or MIMO can still be beneficial by increasing effective SINR for instantaneous rate adjustment cases, they might not significantly impact DCC compared to fixed rate allocation.

\subsubsection{The impact of $\varepsilon_\mathrm{req}$}
Smaller $\varepsilon_\mathrm{req}$ imposes a greater DCC penalty. Furthermore, different rate allocation strategies exhibit significant performance gaps, as illustrated in Fig.~\ref{fig_difeps_full}.

\begin{figure}[!t]
      \centering
      \begin{minipage}[b]{0.48\textwidth}
            \centering
            \subfigure[]{\includegraphics[width = 0.48\textwidth]{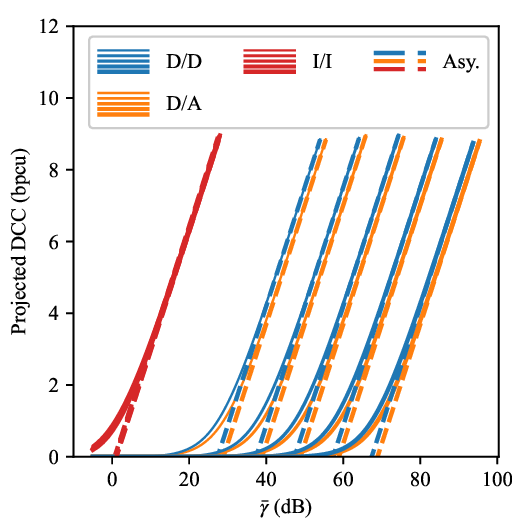}}
            \subfigure[]{\includegraphics[width = 0.48\textwidth]{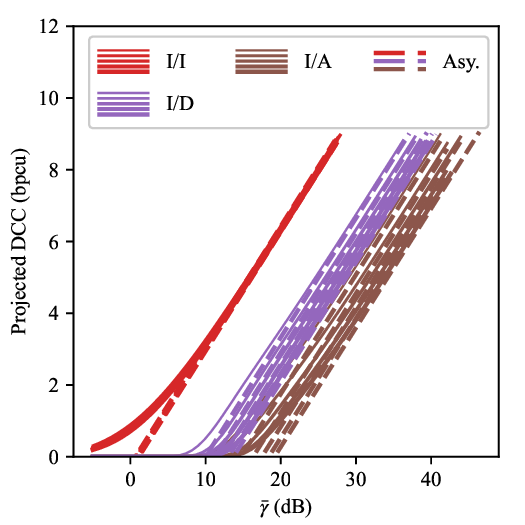}}
      \end{minipage}
      \caption{Projected average DCC against different $\varepsilon_\mathrm{req}$ for different cases.
               (a) For the cases with fixed rate allocation.
               (b) For the cases with instantaneous rate adjustment.
               In both figures, from the most thin to the most bold,
               the solid lines correspond to $\varepsilon_\mathrm{req} = 10^{-5}, 10^{-7}, 10^{-9}, 10^{-11}, 10^{-13}$.}
\label{fig_difeps_full}
\end{figure}
For fixed rate allocation (Fig.~\ref{fig_difeps_full}(a)), Corollary \ref{cor_epsForFixed} is demonstrated, with minimal difference between D/D and D/A cases. However, for instantaneous rate adjustment, the DCC penalty is less sensitive to $\varepsilon_\mathrm{req}$ compared to fixed rate allocation (Fig.~\ref{fig_difeps_full}(b)). While DCC still decreases as the logarithm of $\varepsilon_\mathrm{req}$ decreases, a sublinear trend is observed.

Therefore, for fixed rate allocation, multiple repetitions could boost DCC by relaxing $\varepsilon_\mathrm{req}$, as reception only fails when all replicas fail to decode correctly. Additionally, precise signal power knowledge becomes crucial for meeting extremely stringent reliability requirements.

\section{Conclusion}\label{End}
This paper investigates the impact of interference and signal power precision on achievable rate in industrial wireless communication systems. Analysis of average achievable rate within the FBL regime revealed that per-link interference power information is crucial for system usability and reliability, while instantaneous signal power information enhances performance under stringent QoS requirements.

Applying our theoretical insights, we further explored the impact of fast fading severity and reliability requirements. Our findings demonstrate that if signal power can be instantaneously obtained, the fading severity of interference would noticeably affect the system performance. Otherwise, much more profound effect could be observed for the signal link's fading and the reliability requirement. We anticipate these theorems, results, and insights will advance the development and deployment of 6G wireless networks, particularly by informing network analysis, planning, and resource management strategies within industrial settings.

\appendices
\section{Proof of Theorem \ref{th_avgRate}}\label{proof_th_avgRate}
The average rate is the expectation of the instantaneous rate.
Considering the linearity of the expectation, it is easy to get that
\begin{equation}
      \bar{R} = \mathbb{E}_\gamma \left[\log_2 (1+\gamma)\right] - \frac{Q^{-1}(\varepsilon_\mathrm{req})}{\sqrt{n}} \mathbb{E}_\gamma \left[\sqrt{V(\gamma)}\right].
      \label{eqa_avgRate}
\end{equation}
As stated above,
the "regularized" instantaneous SINR is subject to the F-distribution
with $f = \gamma / \bar{\gamma} \sim \mathrm{F}(m_0, m_I)$.
Thus, \eqref{eqa_avgRate} can be reformulated as
\begin{equation}
      \bar{R} = \underbrace{\mathbb{E}_f \left[\log_2 (1+\bar{\gamma}f)\right]}_{\mathcal{C}_1} - \frac{Q^{-1}(\varepsilon_\mathrm{req})}{\sqrt{n}} \underbrace{\mathbb{E}_f \left[\sqrt{V(\bar{\gamma}f)}\right]}_{\mathcal{C}_2}.
      \label{eqa_avgRateF}
\end{equation}

For $\mathcal{C}_1$, 
first, expanding the expression yields
\begin{equation}
      \begin{aligned}
            \mathcal{C}_1 = \int_0^{+\infty} &\frac{\log_2 (1+\bar{\gamma}f)}{\mathrm{B}(m_0, m_I)} \left(\frac{m_0}{m_I}\right)^{m_0} \times \\
            &f^{m_0-1} \left(1 + \frac{m_0}{m_I}f\right)^{-(m_0 + m_I)} \mathrm{d}f,
      \end{aligned}
\end{equation}
where $\mathrm{B}(a,b)$ is the Beta function\cite[5.12.1]{DLMF}.
Then, using the following series expansion
\begin{equation}
      \log (1+x) = \sum_{i=1}^{\infty} \frac{1}{i} \left(\frac{x}{x+1}\right)^i,
      \label{eqa_lnTaylor}
\end{equation}
it arrives
\begin{equation}
      \begin{aligned}
            &\mathcal{C}_1 = \underbrace{\frac{\log_2e}{\mathrm{B}(m_0, m_I)}\left(\frac{m_0}{m_I}\right)^{m_0}}_{\mathbb{C}} \times \\
            &\int_0^{+\infty} \sum_{k=1}^{\infty} \frac{1}{k} \left(\frac{\bar{\gamma}f}{\bar{\gamma}f+1}\right)^k f^{m_0-1}
            \left(1 + \frac{m_0}{m_I}f\right)^{-(m_0 + m_I)} \mathrm{d}f.
      \end{aligned}
\end{equation}
Next, by swapping the summation and the integration,
and solving the integral, it becomes
\begin{equation}
      \begin{aligned}
            &\mathcal{C}_1 = \mathbb{C} \sum_{k=1}^{\infty} \frac{\bar{\gamma}^k}{k(k+m_0)} \times \\
            &\underbrace{\lim_{f\to\infty} \left[f^{k+m_0} F_1 \left(
                  \begin{array}{c}
                        k+m_0 \\
                        k,m_0+m_I \\
                        k+m_0+1
                  \end{array}
                  ;-\bar{\gamma}f,-\frac{m_0}{m_I}f\right)\right]}_{\mathcal{L}_1},
      \end{aligned}
      \label{eqa_solvedintg}
\end{equation}
where, $F_1\left(\begin{array}{c} a \\ b_1, b_2 \\ c \end{array}; x, y\right)$ is the Appell's $F_1$ function\cite[16.13.1]{DLMF}.
Solving $\mathcal{L}_1$ will involve the analytical continuation of the $F_1$ function,
which could refer to \cite{Colavecchia2001}.
Thus, we have
\begin{equation}
      \begin{aligned}
            \mathcal{L}_1 = &\lim_{f\to\infty} \frac{\Gamma(k+m_0+1)\Gamma(-m_I)}{\Gamma(k+m_0)\Gamma(1-m_I)} (\bar{\gamma})^{-1} \times \\
                                            &\quad\left(1+\frac{m_0}{m_I}f\right)^{-(m_0 + m_I)}(1+\bar{\gamma}f)^{1-k}f^{k+m_0-1} \times \\
                                            &\quad F_1\left(\begin{array}{c}
                                                1 \\
                                                1-k-m_0,m_0+m_I\\
                                                1+m_I
                                            \end{array}
                                                ;-\frac{1}{\bar{\gamma}f}, \frac{\bar{\gamma}-\frac{m_0}{m_I}}{\frac{m_0}{m_I}\bar{\gamma}f+\bar{\gamma}}\right)\\
                                            &+\frac{\Gamma(k+m_0+1)\Gamma(m_I)}{\Gamma(k+m_0+m_I)} (\bar{\gamma})^{m_I-1} \times \\
                                            &\quad\left(1+\frac{m_0}{m_I}f\right)^{-(m_0 + m_I)}(1+\bar{\gamma}f)^{1-k}f^{k+m_0+m_I-1} \times \\
                                            &\quad G_2\left(\begin{array}{c}
                                                1-m_I,m_0+m_I \\
                                                m_I,1-k-m_0-m_I
                                            \end{array}
                                                ;\frac{1}{\bar{\gamma}f}, \frac{(\bar{\gamma}-\frac{m_0}{m_I})f}{1+\frac{m_0}{m_I}f}\right),
      \end{aligned}
      \label{eqa_L1expand}
\end{equation}
where $\Gamma(x)$ is the Gamma function\cite[5.2.1]{DLMF}, $G_2\left(\begin{array}{c} a, a^\prime \\ b, b^\prime\end{array}; x, y\right)$ is the Horn's $G_2$ function\cite[1.3(7)]{Srivastava1985}.
As $f$ goes to the infinity, the first term in \eqref{eqa_L1expand} vanishes,
and the $G_2$ function in the second term yields
\begin{equation}
      \begin{aligned}
            &\lim_{f\to\infty}G_2\left(\begin{array}{c}
                  1-m_I,m_0+m_I \\
                  m_I,1-k-m_0-m_I
               \end{array}
                  ;\frac{1}{\bar{\gamma}f}, \frac{(\bar{\gamma}-\frac{m_0}{m_I})f}{1+\frac{m_0}{m_I}f}\right) \\
            =&{_2}F_1\left(\begin{array}{c}
                  k+m_0,m_0+m_I \\
                  k+m_0+m_I
            \end{array}
                  ;1-\frac{m_0}{m_I\bar{\gamma}}\right),
      \end{aligned}
\end{equation}
where ${_2}F_1\left(\begin{array}{c} a, b \\ c \end{array}; x\right)$ is the Gaussian hypergeometric function\cite[15.2.1]{DLMF}.

Then, $\mathcal{L}_1$ becomes
\begin{equation}
      \begin{aligned}
      \mathcal{L}_1 = &\bar{\gamma}^{-(k+m_0)} \frac{\Gamma(k+m_0+1)\Gamma(m_I)}{\Gamma(k+m_0+m_I)} \times\\
      &{_2}F_1\left(\begin{array}{c}
            k+m_0,m_0+m_I \\
            k+m_0+m_I
      \end{array}
            ;1-\frac{m_0}{m_I\bar{\gamma}}\right).            
      \end{aligned}
      \label{eqa_L1}
\end{equation}
Finally, plugging \eqref{eqa_L1} into \eqref{eqa_solvedintg} will yield
\begin{equation}
      \begin{aligned}
            \mathcal{C}_1 = &\frac{\log_2e}{\mathrm{B}(m_0, m_I)} \left(\frac{m_0}{m_I\bar{\gamma}}\right)^{m_0} \times \\
            &\sum_{k=1}^{\infty}\frac{\Gamma(k+m_0+1)\Gamma(m_I)}{k(k+m_0)\Gamma(k+m_0+m_I)} \times \\
            &\qquad{_2}F_1\left(
                  \begin{array}{c}
                        k+m_0,m_0+m_I \\
                        k+m_0+m_I
                  \end{array}
                  ;1-\frac{m_0}{m_I\bar{\gamma}}\right).
      \end{aligned}
      \label{eq_C1}
\end{equation}

For $\mathcal{C}_2$, as $m_I=0.5$, the channel dispersion coincides with \eqref{eq_AWGNdispersion}. Thus,
\begin{equation}
      \begin{aligned}
            \mathcal{C}_2 = \int_0^{+\infty} &\sqrt{{1-\frac{1}{(1+\bar{\gamma}f)^{2}}}}\times\frac{\log_2 e}{\mathrm{B}(m_0, 0.5)} \left(2m_0\right)^{m_0} \times \\
            &f^{m_0-1} \left(1 + 2m_0f\right)^{-(m_0 + 0.5)} \mathrm{d}f.
      \end{aligned}
\end{equation}
Considering the expansion
\begin{equation}
      \sqrt{{1-\frac{1}{(1+\bar{\gamma}f)^{2}}}} = 1 - \sum_{n=1}^{\infty} \frac{(2n)!}{4^n (n!)^2(2n-1)} (1+\bar{\gamma}f)^{-2n},
\end{equation}
and following the similar deducting routine as that of $\mathcal{C}_1$ will yield
\begin{equation}
      \begin{aligned}
            \mathcal{C}_2 = &\log_2e - \frac{\log_2e}{\mathrm{B}(m_0, 0.5)} \left(\frac{2m_0}{\bar{\gamma}}\right)^{m_0} \times \\
            &\sum_{k=1}^{\infty} \frac{(2k)!\Gamma(m_0+1)\Gamma(0.5+2k)}{4^k (k!)^2 (2k-1)m_0\Gamma(2k+m_0+0.5)} \times \\
            &\qquad{_2}F_1\left(
                  \begin{array}{c}
                        m_0,m_0+0.5 \\
                        2k+m_0+0.5
                  \end{array}
                  ;1-\frac{2m_0}{\bar{\gamma}}\right).
      \end{aligned}
      \label{eq_C2}
\end{equation}

Lengthy as the equations seem, only three special functions (Beta function, Gamma function, and Gaussian hypergeometric function) are utilized, and they are well implemented in computing software. Besides, truncating the infinite summation results in approximation, and the level of precision depends on the number of the terms to be evaluated.

\section{Proof of Theorem \ref{th_asyII}}\label{proof_th_asyII}
The asymptote has the following expression:
\begin{equation}
      \bar{R} = k\bar{\gamma}\mathrm{(dB)} + b,
      \label{eqa_kxpb}
\end{equation}
where $\bar{\gamma}\mathrm{(dB)} = 10 \log_{10} \bar{\gamma}$.
Thus, if the expression of the slope $k$ and the intercept $b$ can be found,
the expression of the asymptote is then completed.

\subsection{The Slope}
The slope can be determined by
\begin{equation}
      k = \lim_{\bar{\gamma}\to\infty} \frac{\bar{R}}{10 \log_{10} \bar{\gamma}}.
      \label{eqa_slope}
\end{equation}
Recalling \eqref{eqa_avgRate}, we have
\begin{equation}
      k = \lim_{\bar{\gamma}\to\infty} \frac{\mathbb{E}_\gamma \left[\log_2 (1+\gamma)\right] - \frac{Q^{-1}(\varepsilon_\mathrm{req})}{\sqrt{n}} \mathbb{E}_\gamma \left[\sqrt{V(\gamma)}\right]}{10 \log_{10} \bar{\gamma}}.
\end{equation}
Note that the channel dispersion function is bounded.
As a result, the expectation is also bounded.
Thus, with $\bar{\gamma}$ goes to the infinity,
this term vanishes:
\begin{equation}
      k = \lim_{\bar{\gamma}\to\infty} \frac{\mathbb{E}_\gamma \left[\log_2 (1+\gamma)\right]}{10 \log_{10} \bar{\gamma}}.
      \label{eqa_IIk1}
\end{equation}

Plugging \eqref{eq_C1} into \eqref{eqa_IIk1},
and using the identity in \cite[15.4.23]{DLMF}, \eqref{eqa_IIk1} becomes
\begin{equation}
      \begin{aligned}
            k = &\frac{\log_2e}{\mathrm{B}(m_0, m_I)} \times \\
            &\lim_{\bar{\gamma}\to\infty} \frac{\sum_{k=1}^{\infty} \frac{1}{k(k+m_0)} \frac{\Gamma(k+m_0+1)\Gamma(m_I)}{\Gamma(k+m_0+m_I)} \frac{\Gamma(k+m_0+m_I)\Gamma(m_0)}{\Gamma(k+m_0)\Gamma(m_0+m_I)}}{10 \log_{10} \bar{\gamma}}.
      \end{aligned}
      \label{eqa_IIk3}
\end{equation}
Further, considering the identity
$\Gamma(z)=\frac{\Gamma(z+1)}{z},$
and the definition of the Beta function $\mathrm{B}(a, b) = \frac{\Gamma(a)\Gamma(b)}{\Gamma(a+b)}$,
we get
\begin{equation}
      \begin{aligned}
            k = \frac{\log_{2}10}{10} \lim_{\bar{\gamma}\to\infty} \frac{\sum_{k=1}^{\infty} \frac{1}{k}}{\log \bar{\gamma}}.
      \end{aligned}
      \label{eqa_IIk4}
\end{equation}
Finally, expanding $\log \bar{\gamma}$ with \eqref{eqa_lnTaylor} yields
\begin{equation}
      \begin{aligned}
            k = \frac{\log_{2}10}{10} \lim_{\bar{\gamma}\to\infty} \frac{\sum_{k=1}^{\infty} \frac{1}{k}}{\sum_{m=1}^{\infty} \frac{1}{m}\left(1 - \frac{1}{\bar{\gamma}}\right)^m} = \frac{\log_{2}10}{10}.
      \end{aligned}
      \label{eqa_IIkfinal}
\end{equation}

\subsection{The Intercept}
With the slope determined,
the intercept, if exists, can be computed by
\begin{equation}
      b = \lim_{\bar{\gamma}\to\infty} [\bar{R} - k (10 \log_{10} \bar{\gamma})].
      \label{eqa_intercept}
\end{equation}

Plugging \eqref{eqa_avgRate} and \eqref{eqa_IIkfinal} into \eqref{eqa_intercept} yields
\begin{equation}
      \begin{aligned}
            b = &\underbrace{\lim_{\bar{\gamma}\to\infty} \left\{\mathbb{E}_\gamma \left[\log_2 (1+\gamma)\right] - \log_2 \bar{\gamma}\right\}}_{f_1} \\
            &- \underbrace{\lim_{\bar{\gamma}\to\infty} \frac{Q^{-1}(\varepsilon_\mathrm{req})}{\sqrt{n}} \mathbb{E}_\gamma \left[\sqrt{V(\gamma)}\right]}_{f_2}.            
      \end{aligned}
      \label{eqa_IIb1}
\end{equation}

For $f_1$, we have
\begin{equation}
      \begin{aligned}
            f_1
            & =\lim_{\bar{\gamma}\to\infty} \mathbb{E}_\gamma \left[\log_2 (1+\gamma)\right] - \log_2 \bar{\gamma} \\
            & =\lim_{\bar{\gamma}\to\infty} \mathbb{E}_{\mathcal{S},\mathcal{I}_E} \left[\log_2 \frac{\mathcal{S}}{\mathcal{I}_E}\right] - \log_2 \bar{\gamma}\\
            & =\lim_{\bar{\gamma}\to\infty} \log_2e \{\mathbb{E}_\mathcal{S} \left[\log \mathcal{S}\right] - \mathbb{E}_{\mathcal{I}_E} \left[\log \mathcal{I}_E\right] - \log \bar{\gamma}\}\\
            & =\log_2e [\psi(m_0) - \log m_0 - \psi(m_I) + \log m_I].
      \end{aligned}
      \label{eqa_IIb11}
\end{equation}

Meanwhile, for $f_2$, we have
\begin{equation}
      \begin{aligned}
            f_2 = \lim_{\bar{\gamma}\to\infty} \frac{Q^{-1}(\varepsilon_\mathrm{req})}{\sqrt{n}} \mathbb{E}_f \left[\sqrt{V(\bar{\gamma}f)}\right].
      \end{aligned}
      \label{eqa_IIb12}
\end{equation}
As $V(\infty) = \frac{(\log_2 e)^2}{2m_I}$, it directly yields
\begin{equation}
      f_2 = \frac{\log_2e}{\sqrt{2m_In}}Q^{-1}(\varepsilon_\mathrm{req}).
      \label{eqa_IIb12final}
\end{equation}

Finally, plugging \eqref{eqa_IIb11} and \eqref{eqa_IIb12final} into \eqref{eqa_IIb1} will result in
\begin{equation}
      b = \log_2e \left[\psi(m_0) - \psi(m_I) + \log \frac{m_I}{m_0} - \frac{Q^{-1}(\varepsilon_\mathrm{req})}{\sqrt{2m_In}}\right].
      \label{eqa_IIbfinal}
\end{equation}

\section{Proof of Lemma \ref{le_asyCD}}\label{proof_le_asyCD}
The asymptote has the same form as \eqref{eqa_kxpb}, with possibly different slope and intercept.
Thus, our goal is to solve them.

\subsection{The Slope}\label{proof_th_asyID_k}
Using the linearization technique introduced in \cite{Makki2014},
\eqref{eq_expRateMeanID} can be approximated as
\begin{equation}
      \varepsilon_\mathrm{req} \approx \omega\sqrt{n} \int_{\gamma_1}^{\gamma_2} F_\gamma(x)\mathrm{d}x,
      \label{eqa_IDkepsLin}
\end{equation}
where $F_\gamma(x)$ is the CDF of $\gamma$,
$$\omega = \left[2^{R}\sqrt{2\pi V(2^{R}-1)}\right]^{-1},$$ and
$$\gamma_{1,2} = 2^{R} \left[ 1 \mp \sqrt{\frac{\pi V(2^{R}-1)}{2n}} \right] -1.$$

Further, as $\lim_{\bar{\gamma}\to\infty} R = \infty$,
we have $$\lim_{\bar{\gamma}\to\infty} V(2^{R}-1) = V(\infty) = \frac{(\log_2e)^2}{2m_I}.$$

Thus, taking the limits on both sides of \eqref{eqa_IDkepsLin},
and let $z = x / \bar{\gamma}$,
we have
\begin{equation}
      \begin{aligned}
            \varepsilon_\mathrm{req} = \lim_{\bar{\gamma}\to\infty}& \log 2\sqrt{\frac{m_In}{\pi}}\left[ \frac{2^{R}}{\bar{\gamma}} \right]^{-1}\times \\
            &\int_{\frac{2^{R}}{\bar{\gamma}}\left[1-\log_2e\sqrt{\frac{\pi}{4m_In}}\right]}^{\frac{2^{R}}{\bar{\gamma}}\left[1+\log_2e\sqrt{\frac{\pi}{4m_In}}\right]} F_\gamma(z)\mathrm{d}z.
      \end{aligned}
      \label{eqa_IDkepsLim}
\end{equation}
For simplicity, we denote $G_+ = 1+\log_2e\sqrt{\frac{\pi}{4m_In}}$,
and $G_- = 1-\log_2e\sqrt{\frac{\pi}{4m_In}}$.
Let $a = \frac{2^{R}}{\bar{\gamma}}$, it yields
\begin{equation}
      \varepsilon_\mathrm{req} = \lim_{\bar{\gamma}\to\infty} \log 2 \sqrt{\frac{m_In}{\pi}}a^{-1}
      \int_{aG_-}^{aG_+}
      F_\gamma(z)\mathrm{d}z.
      \label{eqa_IDkepsA}
\end{equation}
One can observe that when $\bar{\gamma}$ goes to infinity,
$a$ is a function of $\varepsilon_\mathrm{req}$.
As $\varepsilon_\mathrm{req}$ is a constant,
$a$ is also invariable.
Simply rearranging the expression of $a$ will get
\begin{equation}
      R = \log_2 a\bar{\gamma} = \log_2 \bar{\gamma} + \log_2 a.
      \label{eqa_IDkR}
\end{equation}
Thus,
\begin{equation}
      k = \lim_{\bar{\gamma}\to\infty} \frac{R}{10\log_{10} \bar{\gamma}}
      = \lim_{\bar{\gamma}\to\infty} \frac{\log_2 \bar{\gamma}}{10\log_{10} \bar{\gamma}}
      = \frac{\log_{2}10}{10}.
      \label{eqa_IDkfinal}
\end{equation}

Additionally, plugging \eqref{eqa_IDkR} and \eqref{eqa_IDkfinal} into \eqref{eqa_kxpb},
we will immediately get
\begin{equation}
      b = \lim_{\bar{\gamma}\to\infty} \log_2 a.
      \label{eqa_IDbC3}
\end{equation}

\subsection{The Intercept}

As $\gamma$ is subject to the inverse-gamma distribution,
taking the CDF of the distribution into \eqref{eqa_IDkepsA} yields
\begin{equation}
      \begin{aligned}
            \varepsilon_\mathrm{req} &= \lim_{\bar{\gamma}\to\infty} \frac{m_Ia^{-1}\log 2}{\Gamma(m_I)}
                              \sqrt{\frac{m_In}{\pi}}
                              \int_{aG_-m_I^{-1}}^{aG_+m_I^{-1}}
                              \Gamma\left(m_I,x^{-1}\right)\mathrm{d}x \\
                        &\approx \lim_{\bar{\gamma}\to\infty} \frac{m_Ia^{-1}\log 2}{\Gamma(m_I)}\sqrt{\frac{m_In}{2\pi}}
                        \int_{m_I\left(aG_+\right)^{-1}}^{m_I\left(aG_-\right)^{-1}}
                        x^{m_I-3}e^{-x}\mathrm{d}x,
      \end{aligned}
      \label{eqa_IDbint}
\end{equation}
where $\Gamma(x, y)$ is the upper incomplete gamma function.

Unfortunately, the integral is still intractable w.r.t. $a$.
And utilizing the mean-value theorem helps us get rid of the integral:
\begin{equation}
      \varepsilon_\mathrm{req} = \lim_{\bar{\gamma}\to\infty} \frac{m_I^2a^{-2}}{G_+G_-\Gamma(m_I)} \varphi^{m_I-3}e^{-\varphi},
      \label{eqa_IDbmeanvalue}
\end{equation}
where $\varphi\in\left[\frac{m_I}{aG_+}, \frac{m_I}{aG_-}\right]$.
However, $\varphi$ is uncertain.
Considering the reliability,
taking $\varphi^* = \frac{m_I}{aG_+}$ will result in a lower bound.
Plugging $\varphi^*$ into \eqref{eqa_IDbmeanvalue}, after some rearrangement, for $m_I \ne 1$,
we get
\begin{equation}
      \lim_{\bar{\gamma}\to\infty} a = \lim_{\bar{\gamma}\to\infty} -\frac{m_I}{(m_I-1)G_+W\left(-\frac{\sqrt[m_I-1]{\mathbb{A}}}{m_I-1}\right)},
      \label{eqa_IDba}
\end{equation}
where $\mathbb{A} = \varepsilon_\mathrm{req} G_+^{-1} G_- \Gamma(m_I)$.
While for $m_I = 1$, the rearrangement yields a simpler result
\begin{equation}
      \lim_{\bar{\gamma}\to\infty} a = \lim_{\bar{\gamma}\to\infty} \left[ G_+ \left( \log G_+ - \log G_- - \log\varepsilon_\mathrm{req} \right) \right]^{-1}.
      \label{eqa_IDba-1}
\end{equation}

Finally, for $m_I \ne 1$, taking \eqref{eqa_IDba} into \eqref{eqa_IDbC3} results in $b=\mathcal{C}_4$ with
\begin{equation}
      \begin{aligned}
            \mathcal{C}_4 \approx &\log_2m_I - \log_2G_+ - \log_2\left[-\frac{W\left(-\frac{\sqrt[m_I-1]{\mathbb{A}}}{m_I-1}\right)}{m_I-1}\right],
      \end{aligned}
      \label{eq_C3}
\end{equation}
where
\begin{equation}
      \mathbb{A} = \varepsilon_\mathrm{req}\Gamma(m_I)\left[1-\frac{(\log_2e)^2\pi}{4m_In}\right]G_+^{-2},
\end{equation}
and $W(x)$ is the Lambert $W$ function\cite[4.13]{DLMF}.

Whereas for $m_I = 1$, taking \eqref{eqa_IDba-1} into \eqref{eqa_IDbC3} results in
\begin{equation}
      \begin{aligned}
            \mathcal{C}_4 \approx &- \log_2G_+ - \log_2\left\{\log\left[\frac{\varepsilon_\mathrm{req} G_-}{G_+}\right]\right\},
      \end{aligned}
      \label{eq_C3-1}
\end{equation}
which concludes the proof.

\section{Proof of Theorem \ref{th_asyDD}}\label{proof_th_asyDD}
The asymptote has the same form of expression as \eqref{eqa_kxpb}.
Using Corollary \ref{cor_slopeInvar},
the asymptote becomes:
\begin{equation}
      \bar{R} = \frac{\log_2 10}{10}\bar{\gamma}\mathrm{(dB)} + b.
      \label{eqa_DDkxpb}
\end{equation}

Also, as Corollary \ref{cor_slopeInvar} suggests
that the derivation in Appendix \ref{proof_th_asyID_k} is distribution-independent,
we can directly utilize \eqref{eqa_IDkepsA}
to continue this proof.\footnote{We also have $G_+ = 1+\log_2e\sqrt{\frac{\pi}{4m_In}}$, and $G_- = 1-\log_2e\sqrt{\frac{\pi}{4m_In}}$ in this appendix.}
Thus, we get
\begin{equation}
      \begin{aligned}
            \varepsilon_\mathrm{req} &= \lim_{\bar{\gamma}\to\infty} \log 2 \sqrt{\frac{m_In}{\pi}}a^{-1}
            \int_{aG_-}^{aG_+}
            F_{\hat{\gamma}}(x)\mathrm{d}x \\
            &= \lim_{\bar{\gamma}\to\infty} \frac{\log 2}{am_0\mathrm{B}(m_0, m_I)} \sqrt{\frac{m_In}{\pi}} \times \\
            &\int_{aG_-}^{aG_+}
            \left(\frac{m_0x}{m_0x+m_I}\right)^{m_0} \left(\frac{m_I}{m_0x+m_I}\right)^{m_I-1} \times \\
            &{_2}F_1\left(
                  \begin{array}{c}
                        1,1-m_I\\
                        m_0+1
                  \end{array}
                  ;-\frac{m_0}{m_I}x\right)
            \mathrm{d}x.
      \end{aligned}
      \label{eqa_DDepsA}
\end{equation}

The fact is that, from \eqref{eqa_DDepsA},
$\varepsilon_\mathrm{req}$ equals to the mean value of the CDF within the given interval.
As the derivative of the CDF, i.e. the PDF,
of $\hat{\gamma}$ is always positive within its domain,
the CDF is strictly monotonically increasing.
Thus, with more stringent reliability requirement,
the bounds of the integration interval will be increasingly close to $0$.
As a result, we have the following approximations
\begin{equation}
      \begin{aligned}
            \varepsilon_\mathrm{req} &\approx \lim_{\bar{\gamma}\to\infty} \frac{m_0^{m_0-1}\log 2}{a\mathrm{B}(m_0, m_I)} \sqrt{\frac{m_In}{\pi}} \times \\
            &\int_{aG_-}^{aG_+}
            x^{m_0}(m_0x+m_I)^{1-m_0-m_I}\mathrm{d}x \\
            &\approx \lim_{\bar{\gamma}\to\infty} \frac{a^{m_0}m_0^{m_0-1}\log 2}{m_I^{m_0}\mathrm{B}(m_0, m_I)}\sqrt{\frac{m_In}{\pi}} 
            \underbrace{\left\{G_+^{m_0+1}-G_-^{m_0+1}\right\}}_{\mathbb{S}}.
      \end{aligned}
      \label{eqa_DDbepsfinal}
\end{equation}

Rearranging \eqref{eqa_DDbepsfinal}, we get
\begin{equation}
      \lim_{\bar{\gamma}\to\infty} a = \frac{m_I}{m_0}\left[\frac{\varepsilon_\mathrm{req} m_0(m_0+1)\mathrm{B}(m_0, m_I)}{\mathbb{S}\log 2\sqrt{\frac{m_In}{\pi}}}\right]^{\frac{1}{m_0}}.
      \label{eqa_DDba}
\end{equation}
Finally, taking \eqref{eqa_DDba} into \eqref{eqa_IDbC3} yields $b=\mathcal{C}_5$ with
\begin{equation}
      \mathcal{C}_5 = \log_2 \frac{m_I}{m_0} 
      + m_0^{-1}\log_2 \frac{\varepsilon_\mathrm{req} m_0 (m_0+1)\mathrm{B}(m_0, m_I)\sqrt{\pi}}{\mathbb{S}\sqrt{m_In}\log2},
      \label{eq_C5}
\end{equation}
where
\begin{equation}
      \mathbb{S} = G_+^{m_0+1} - G_-^{m_0+1}.
      \label{eq_C5S}
\end{equation}

\ifCLASSOPTIONcaptionsoff
  \newpage
\fi
\bibliographystyle{IEEEtran}
\balance
\bibliography{biblib, extralib, fblbib, fblbib2}
\end{document}